\documentclass[11pt]{article}

\usepackage{diagbox}
\usepackage{multirow}

\usepackage[tbtags]{mathtools}
\usepackage{amsmath,amssymb,amsfonts,amsthm} 
\usepackage{graphicx}
\usepackage{verbatim}
\usepackage{framed}
\usepackage{bbm}
\usepackage{fullpage}
\usepackage{color}
\usepackage{xspace}
\usepackage[ruled]{algorithm}
\usepackage{algpseudocode}
\usepackage{algorithmicx}
\floatname{algorithm}{Protocol}
\usepackage{times}
\usepackage{thmtools}
\usepackage{thm-restate}
\usepackage[margin=1in]{geometry}

\usepackage[pagebackref=true]{hyperref}
\hypersetup{
    unicode=false,          
    colorlinks=true,        
    linkcolor=red,          
    citecolor=blue,        
    filecolor=magenta,      
    urlcolor=cyan           
}

\newtheorem{theorem}{Theorem}

\newtheorem{lem}{Lemma}
\newtheorem{lemma}[lem]{Lemma}

\newtheorem{cor}{Corollary}

\newtheorem{claim}{Claim}

\newenvironment{proof-sketch}{\noindent{\bf Sketch of Proof}\hspace*{1em}}{\qed\bigskip}
\newenvironment{proof-idea}{\noindent{\bf Proof Idea}\hspace*{1em}}{\qed\bigskip}
\newenvironment{proof-of-lemma}[1]{\noindent{\bf Proof of Lemma #1}\hspace*{1em}}{\qed\bigskip}
\newenvironment{proof-attempt}{\noindent{\bf Proof Attempt}\hspace*{1em}}{\qed\bigskip}

\newcommand{\F}{{\mathbb F}}

\newcommand{\calc}{{\cal C}}
\newcommand{\calp}{{\cal P}}
\newcommand{\call}{{\cal L}}
\newcommand{\Indi}{\mathbb{I}}
\newcommand{\bbt}{\mathbb{T}}

\newcommand{\supp}{\mathrm{supp}}
\newcommand{\bool}{\{0, 1\}}
\newcommand{\barcals}{\bar{\cal S}}
\newcommand{\1}{\mathbf{1}}
\newcommand{\0}{\mathbf{0}}

\newcommand{\E}{\mathbb E}

\newcommand{\por}{P_{\vee}}

\begin{document}

\title{LP/SDP Hierarchy Lower Bounds \\ for Decoding Random LDPC Codes}
\author{
 Badih Ghazi\footnote{Computer Science and Artificial Intelligence Laboratory, Massachusetts Institute of Technology, Cambridge MA 02139. Supported in part by NSF STC Award CCF 0939370 and NSF award number CCF-1217423.} \\ MIT \\ \texttt{badih@mit.edu} \and 
 Euiwoong Lee\footnote{Computer Science Department, Carnegie Mellon University, Pittsburgh, PA 15213. Supported by a Samsung Fellowship, US-Israel BSF grant 2008293, and NSF CCF-1115525. Most of this work was done while visiting Microsoft Research New England.} \\ CMU\\ \texttt{euiwoonl@cs.cmu.edu} }

\maketitle
\thispagestyle{empty} 

\begin{abstract}

Random $(d_v,d_c)$-\emph{regular} LDPC codes (where each variable is involved in $d_v$ parity checks and each parity check involves $d_c$ variables) are well-known to achieve the Shannon capacity of the binary symmetric channel (for sufficiently large $d_v$ and $d_c$) under exponential time decoding. However, polynomial time algorithms are only known to correct a much smaller fraction of errors. One of the most powerful polynomial-time algorithms with a formal analysis is the LP decoding algorithm of Feldman et al. which is known to correct an $\Omega(1/d_c)$ fraction of errors. In this work, we show that fairly powerful extensions of LP decoding, based on the Sherali-Adams and Lasserre hierarchies, fail to correct much more errors than the basic LP-decoder. In particular, we show that:

\begin{itemize}
\item For any values of $d_v$ and $d_c$, a linear number of rounds of the Sherali-Adams LP hierarchy cannot correct more than an $O(1/d_c)$ fraction of errors on a random $(d_v,d_c)$-regular LDPC code.
\item For any value of $d_v$ and infinitely many values of $d_c$, a linear number of rounds of the Lasserre SDP hierarchy cannot correct more than an $O(1/d_c)$ fraction of errors on a random $(d_v,d_c)$-regular LDPC code.
\end{itemize}

Our proofs use a new \emph{stretching} and \emph{collapsing} technique that allows us to leverage recent progress in the study of the limitations of LP/SDP hierarchies for Maximum Constraint Satisfaction Problems (Max-CSPs). The problem then reduces to the construction of special \emph{balanced pairwise independent distributions} for Sherali-Adams and special \emph{cosets of balanced pairwise independent subgroups} for Lasserre. Our (algebraic) construction for the Lasserre hierarchy is based on designing sets of points in $\F_q^d$ (for $q$ any power of $2$ and $d = 2,3$) with special hyperplane-incidence properties --- constructions that may be of independent interest. An intriguing consequence of our work is that \emph{expansion} seems to be both the \emph{strength} and the \emph{weakness} of random regular LDPC codes.

Some of our techniques are more generally applicable to a large class of Boolean CSPs called Min-Ones. In particular, for $k$-Hypergraph Vertex Cover, we obtain an improved integrality gap of $k-1-\epsilon$ that holds after a \emph{linear} number of rounds of the Lasserre hierarchy, for any $k = q+1$ with $q$ an arbitrary prime power. The best previous gap for a linear number of rounds was equal to $2-\epsilon$ and due to Schoenebeck.

\end{abstract}

\newpage

\tableofcontents

\thispagestyle{empty} 
\newpage
\setcounter{page}{1}

\section{Introduction}

Low-density parity-check (LDPC) codes are a class of linear error correcting codes originally introduced by Gallager \cite{gallager1962low} and that have been extensively studied in the last decades. A $(d_v,d_c)$-LDPC code of block length $n$ is described by a parity-check matrix $H \in \mathbb{F}_2^{m \times n}$ (with $m \le n$) having $d_v$ ones in each column and $d_c$ ones in each row. It can be also represented by its bipartite {\em parity-check graph} $(L \cup R, E)$ where $L$ corresponds to the columns of $H$, $R$ corresponds to the rows of $H$, and $(u, v) \in E$ if and only if $H_{v, u} = 1$.
For a comprehensive treatment of LDPC codes, we refer the reader to the book of Richardson and Urbanke \cite{richardson2008modern}. In many studies of LDPC codes, \emph{random LDPC codes} have been considered. For instance, Gallager studied in his thesis the distance and decoding-error probability of an ensemble of random $(d_v,d_c)$-LDPC codes. Random $(d_v,d_c)$-LDPC codes were further studied in several works (e.g., \cite{sipser1994expander, mackay1999good, richardson2001capacity, MB01, di2002finite, litsyn2002ensembles, kudekar2012spatially}). The reasons why random $(d_v,d_c)$-LDPC codes have been of significant interest are their nice properties, their tendency to simplify the analysis of the decoding algorithms and the potential lack of known explicit constructions for properties satisfied by random codes.

One such nice property that is exhibited by random $(d_v, d_c)$-LDPC codes is the \emph{expansion} of the underlying parity-check graph. Sipser and Spielman \cite{sipser1994expander} exploited this expansion in order to give a linear-time decoding algorithm correcting a constant fraction of errors (for $d_v ,d_c = O(1)$). More precisely, they showed that if the underlying graph has the property that every subset of at most $\delta n$ variable nodes expands by at least a factor of $3 d_v/4$, then their decoding algorithm can correct an $\Omega(\delta)$ fraction of errors in linear-time. Since, with high probability, a random $(d_v,d_c)$-LDPC code satisfies this expansion property for some $\delta = \Omega(1/d_c)$, this implies that the linear-time decoding algorithm of Sipser-Spileman corrects $\Omega(1/d_c)$-errors on a random $(d_v,d_c)$-LDPC code. A few years after the work of Sipser-Spielman, Feldman, Karger and Wainwright \cite{feldman2005using,feldman2003decoding} introduced a decoding algorithm that is based on a simple linear programming (LP) relaxation, and a later paper by Feldman, Malkin, Servedio, Stein and Wainwright \cite{feldman2007lp} showed that when the underlying parity-check graph has the property that every subset of at most $\delta n$ variable nodes expands by a factor of at least $2d_v/3 + \Omega(1)$, the linear program of Feldman-Karger-Wainwright corrects $\Omega(\delta)$ errors. Again, since with high probability, a random $(d_v,d_c)$-LDPC code satisfies this expansion property for some $\delta = \Omega(1/d_c)$, this means that the LP of~\cite{feldman2005using} corrects $\Omega(1/d_c)$-errors on a random $(d_v,d_c)$-LDPC code.

However, the fraction of errors that is corrected by the Sipser-Spielman algorithm and the LP relaxation of~\cite{feldman2005using} (which is $O(1/d_c)$) can be much smaller than the best possible: in fact, \cite{gallager1962low} (as well as \cite{MB01}) showed that for a random $(d_v,d_c)$-LDPC code, the exponential-time nearest-neighbor Maximum Likelihood (ML) algorithm corrects close to $H_b^{-1}(d_v/d_c)$ probabilistic errors, which by Shannon's channel coding theorem is the best possible\footnote{More precisely, the fraction of errors corrected by the ML decoder is bounded below $H_b^{-1}(d_v/d_c)$ for fixed $d_c$ but gets arbitrarily close to $H_b^{-1}(d_v/d_c)$ as $d_c$ gets larger.}. Note that, for example, if we set the ratio $d_v/d_c$ to be a small constant and let $d_c$ grow, then the fraction of errors that is corrected by the Sipser-Spielman algorithm and the LP relaxation of Feldman et al. decays to $0$ with increasing $d_c$, whereas the maximum information-theoretically possible fraction is a fixed absolute constant!\footnote{In fact, not only is the fraction of \emph{probabilistic} errors that is corrected by the ML decoder an absolute constant, but so is the fraction of \emph{adversarial} errors \cite{gallager1962low,burshtein2004asymptotic}. More precisely, for say $d_v = 0.1 d_c$, Theorem $11$ of~\cite{burshtein2004asymptotic} implies that the \emph{minimum distance} of a random $(d_v,d_c)$-regular LDPC code is at least an absolute constant and it approaches the Gilbert-Varshamov bound for rate $R = 1-d_v/d_c = 0.9$ as $d_c$ gets larger.} The belief propagation (BP) algorithm also suffers from the same limitation \cite{burshtein2002bounds,kudekar2012spatially}. In fact, there is no known polynomial-time algorithm that approaches the information-theoretic limit for random $(d_v,d_c)$-regular LDPC codes. \footnote{We point out that for some ensembles of \emph{irregular} LPDC codes \cite{richardson2001design} as well as for the recently studied \emph{spatially-coupled} codes \cite{kudekar2012spatially}, belief propagation is known to have better properties. In this paper, our treatment is focused on random \emph{regular} LDPC codes.}

In the areas of combinatorial optimization and approximation algorithms, hierarchies of linear and semidefinite programs such as the Sherali-Adams~\cite{SA90} and the Lasserre~\cite{Lasserre01} hierarchies recently gained significant interest. Given a base LP relaxation, such hierarchies tighten it into sequences of convex programs where the convex program corresponding to the $r$th round in the sequence can be solved in time $n^{O(r)}$ and yields a solution that is ``at least as good'' as those obtained from previous rounds in the sequence. For an introduction and comparison of those LP and SDP hierarchies, we refer the reader to the work of Laurent~\cite{laurent2003comparison} where it is also shown that the Lasserre hierarchy is at least as strong as the Sherali-Adams hierarchy. 

Inspired by the Sherali-Adams hierarchy, Arora, Daskalakis and Steurer~\cite{arora2012message} improved the best known fraction of correctable probabilistic errors by the LP decoder (which was previously achieved by Daskalakis et al.~\cite{daskalakis2008probabilistic}) for some range of values of $d_v$ and $d_c$. Both Arora et al.~\cite{arora2012message} and the original work of Feldman et al.~\cite{feldman2005using,feldman2003decoding}  asked whether tightening the base LP using linear or semidefinite hierarchies can improve its performance, potentially approching the information-theoretic limit. More precisely, in all previous work on LP decoding of error-correcting codes, the base LP decoder of Feldman et al. succeeds in the decoding task if and only if the transmitted codeword is the unique optimum of the relaxed polytope with the objective function being the (normalized) $l_1$ distance between the received vector and a point in the polytope. On the other hand, the decoder is considered to fail whenever there is an optimal non-integral vector\footnote{Such an optimal non-integral vector is called a ``pseudocodeword'' in the LP-decoding literature.}. The hope is that adding linear and semidefinite constraints will help ``prune'' non-integral optima, thereby improving the fraction of probabilistic errors that can be corrected.

In this paper, we prove the first lower bounds on the performance of the Sherali-Adams and Lasserre hierarchies when applied to the problem of decoding random $(d_v,d_c)$-LDPC codes. 
Throughout this paper, by a random $(d_v,d_c)$-LDPC code, we mean one whose parity-check graph is drawn from the following ensemble that was studied in numerous previous works (e.g., \cite{sipser1994expander,richardson2001capacity,MB01,litsyn2002ensembles,burshtein2004asymptotic,kudekar2012spatially}) and is very close to the ensemble that was originally suggested by Gallager \cite{gallager1962low}. Set $M := n d_v = m d_c$ where $n$ is the block length and $m$ is the number of constraints. Assign $d_v$ (resp. $d_c$) sockets to each of $n$ (resp. $m$) vertices on the left (resp. right) and number them $1, \dots , M$ on each side. Sample a permutation $\pi : \{1, \dots , M \} \rightarrow \{1, \dots , M \}$ uniformly at random, and connect the $i$-th socket on the left to the $\pi(i)$-th socket on the right. Place an edge betwen variable $i$ and constraint $j$ if and only if there is an odd number of edges between the sockets corresponding to $i$ and those corresponding to $j$. 
Our main results can be stated as follows:

\begin{theorem}[Lower bounds in the Sherali-Adams hierarchy]
\label{thm:sa_ldpc}
For any $d_v$ and $d_c \geq 5$, there exists $\eta > 0$ (depending on $d_c$) such that a random $(d_v, d_c)$-LDPC code satisfies the following with high probability: 
for any received vector, 
there is a fractional solution to the $\eta n$ rounds of the Sherali-Adams hierarchy of value $1 / (d_c - 3)$ (for odd $d_c$) or $1 / (d_c - 4)$ (for even $d_c$).
Consequently, $\eta n$ rounds cannot decode more than $a \approx 1/d_c$ fraction of errors.
\end{theorem}

\begin{theorem}[Lower bounds in the Lasserre hierarchy]
\label{thm:lasserre_ldpc}
For any $d_v$ and $d_c = 3\cdot 2^i + 3$ with $i \geq 1$, there exists $\eta > 0$ (depending on $d_c$) such that a random $(d_v, d_c)$-LDPC code satisfies the following with high probability: 
for any received vector, 
there is a fractional solution to the $\eta n$ rounds of the Lasserre hierarchy of value $3 / (d_c - 3)$.
Consequently, $\eta n$ rounds cannot decode more than $a \approx 3/d_c$ fraction of errors.
\end{theorem}

We note that Theorems~\ref{thm:sa_ldpc} and~\ref{thm:lasserre_ldpc} hold, in particular, for random errors. We point out that as in all previous work on LP decoding of error-correcting codes, Theorems~\ref{thm:sa_ldpc} and~\ref{thm:lasserre_ldpc} assume that a decoder based on a particular convex relaxation succeeds in the decoding task if and only if the transmitted codeword is the unique optimum of the convex relaxation. Note that the decoder based on the LP (resp. SDP) corresponding to $n$ rounds of the Sherali-Adams (resp. Lasserre) hierarchy is the nearest-neighbor maximum likelihood (ML) decoder.

We note that our LP/SDP hierarchy $O(1/d_c)$ lower bounds for random LDPC codes hold, in particular, for any check-regular code with good check-to-variable expansion.
Moreover, the fact that the base LP corrects $\Omega(1/d_c)$ errors follows from the (variable-to-check) expansion of random LDPC codes\footnote{We note that Feldman et al. \cite{feldman2007lp} first proved that LP decoding corrects $\Omega(1/d_c)$ on expanding graphs. Their proof was recently simplified by Viderman \cite{viderman2013lp} who also slightly relaxed the expansion requirements. Both works assumed that all variable nodes have the same degree but the proof readily extends to the case where variable nodes can have degree either $d_v$ or $d_v-2$, which is the typical case for random $(d_v,d_c)$-LDPC codes.}. In that respect, it is intriguing that expansion constitutes both the \emph{strength} and the \emph{weakness} of random LDPC codes.

Some of our techniques are more generally applicable to a large class of Boolean Constraint Satisfaction Problems (CSPs) called Min-Ones where the goal is to satisfy each of a collection of constraints while minimizing the number of variables that are set to $1$. In particular, we obtain improved integrality gaps in the Lasserre hierarchy for the $k$-uniform Hypergraph Vertex Cover ($k$-HVC) problem. The $k$-HVC problem is known to be NP-hard to approximate within a factor of $k - 1 - \epsilon$~\cite{DGKR05}. This reduction would give the same integrality gap only for some \emph{sublinear} number of rounds of the Lasserre hierarchy, whereas the best integrality gap for a linear number of rounds remains at $2-\epsilon$~\cite{Schoenebeck08}.
We prove that an integrality gap of $k - 1 - \epsilon$ still holds after a \emph{linear} number of rounds, for any $k = q+1$ with $q$ an arbitrary prime power.

\begin{theorem} 
\label{thm:hvc}
Let $k = q+1$ where $q$ is any prime power. For any $\epsilon > 0$, there exist $\beta, \eta > 0$ (depending on $k$) such that a random $k$-uniform hypergraph with $n$ vertices and $m = \beta n$ edges, simultaneously satisfies the following two conditions with high probability. 
\begin{itemize}
\item The integral optimum of $k$-HVC is at least $(1 - \epsilon)n$.
\item There is a solution to the $\eta n$ rounds of the Lasserre hierarchy of value $\frac{1}{k - 1}n$. 
\end{itemize}
\end{theorem}

\subsection{Proof Techniques}

The LP of Feldman et al. \cite{feldman2005using,feldman2003decoding} is a relaxation of the Nearest Codeword problem, where given a binary linear code (represented by its parity-check matrix or graph) and a received vector, the goal is to find the codeword that is closest to it in Hamming distance. The Nearest Codeword problem can be viewed as a particular case of a variant of Constraint Satisfaction Problems (CSPs) called Min-Ones, where the goal is to find an assignment that satisfies all constraints while minimizing the number of ones in the assignment (see \cite{KSTW01} for more on Min-Ones problems). In this Min-Ones view, each codeword bit corresponds to a binary variable that the decoder should decide whether to flip or not.

Recently, there has been a significant progress in understanding the limitations of LP and SDP hierarchies for CSPs (e.g.,~\cite{GMT09, Schoenebeck08, Tulsiani09, Chan13}); in these works, the focus was on a different variant of CSPs called Max-CSPs, where the goal is to find an assignment maximizing the number of satisfied constraints. These results construct fractional solutions satisfying all constraints and that are typically \emph{balanced} in that any coordinate of the assignment is set to $1$ with probability $1/2$ in the case of a binary alphabet. Therefore, they yield a fractional solution where half the variables are fractionally flipped.

In order to construct a fractional solution with a smaller number of (fractionally) flipped variables, we introduce the technique of \emph{stretching} and \emph{collapsing} the domain. Given an instance of the Nearest Codeword problem, we stretch the domain into a finite set $G$ via a map $\phi: G \to \{0,1\}$. The new CSP instance has the same set $V$ of variables but each variable now takes values in $G$ (as opposed to $\{0,1\}$). A constraint in the new instance on variables $(v_1,\dots,v_k)$ is satisfied by an assignment $f: V \to G$ if and only if it is satisfied in the original instance by the assignment $\phi \circ f: V \to \{0,1\}$. Assume that the map $\phi$ satisfies $|\phi^{-1}(1)| = 1$ and that the previous results for Max-CSPs yield a fractional solution over alphabet $G$ such that each variable $v$ takes any particular value $g \in G$ with probability $1/|G|$. If we can transform this fractional solution into one for the original instance by collapsing $\phi^{-1}(i)$ back to $i$ for every $i \in \{0,1\}$, we would get a fractional solution to the original (binary) instance of the Nearest Codeword problem with value $1/|G|$. In Section~\ref{sec:solns_des_struct}, we show that this stretching and collapsing idea indeed works. This technique can be generalized to any Min-Ones problem (e.g., $k$-HVC).

To apply the known constructions for Max-CSPs between our stretching and collapsing steps, we need to construct special structures that are required by those results. For the Sherali-Adams hierarchy in the case of the Nearest Codeword problem, we need to construct two \emph{balanced pairwise independent distributions} on $G^k$: one supported only on vectors with an \emph{even} number of $0$ coordinates and the other supported only on vectors with an \emph{odd} number of $0$ coordinates.\footnote{Here, we are assuming WLOG that $0 \in G$. In fact, we can consider any fixed element of the set $G$.} For the Lasserre hierarchy, we need to construct two \emph{cosets of balanced pairwise independent subgroups}: one supported only on vectors with an even number of $0$ coordinates and the other supported only on vectors with an odd number of $0$ coordinates.

Constructing the desired balanced pairwise independent distributions in the Sherali-Adams hierarchy can be done by setting up systems of linear equations (one variable for each allowed vector $(x_1,\dots,x_k)$ modulo symmetry) and checking that the resulting solution yields a valid probability distribution (see Section~\ref{subsec:sa} for more details). Constructing the desired cosets of balanced pairwise independent subgroups in the Lasserre hierarchy is more involved and our algebraic construction is based on designing sets of points in $\F_q^d$ (for $q$ any power of two and $d = 2,3$) with special hyperplane-incidence properties. One example is the construction (for every power $q$ of $2$) of a subset $E$ of $q+2$ points in $\mathbb{F}_q^2$ containing the origin and such that every line in the $\mathbb{F}_q^2$-plane contains either $0$ or $2$ points in $E$. See Section~\ref{subsec:lasserre} for more details.

Finally, random $(d_v,d_c)$-LDPC codes typically have check nodes with slightly different degrees whereas in the CSP literature, it is common to assume that all the constraints contain the same number of variables. Since our algebraic constructions of cosets of balanced pairwise independent subgroups for Lasserre hold only for specific arity values, we need an additional technique to obtain the required predicates for both arity $d_c$ and arity $d_c-2$ (which are with high probability the two possible check-degrees in a random $(d_v,d_c)$-LDPC code). We construct such predicates by taking the \emph{direct-sums} of pairs and triples of previously constructed cosets, at the expense of multiplying the value of the fractional solution by an absolute constant.

\subsection{Organization}
Section~\ref{sec:preliminaries} provides background on the problems and hierarchies that we study in this paper. Section~\ref{sec:solns_des_struct} introduces the \emph{stretching} and \emph{collapsing} technique and shows how to leverage previous results for Max-CSPs to reduce our problem to the construction of special distributions and cosets. This general result holds for any Min-Ones problem. Section~\ref{sec:ldpc} provides the desired constructions for the problem of decoding random $(d_v,d_c)$-LDPC codes, proving Theorem~\ref{thm:sa_ldpc} in Section~\ref{subsec:sa} and Theorem~\ref{thm:lasserre_ldpc} in Section~\ref{subsec:lasserre}. The proof of Theorem~\ref{thm:hvc} about $k$-Hypergraph Vertex Cover can be found in Appendix~\ref{subsec:hvc}.

\section{Preliminaries}
\label{sec:preliminaries}
\paragraph{Constraint Satisfaction Problems (CSPs) and Min-Ones.}
Fix a finite set $G$. 
Let  $\calp = \{ P_1, \dots , P_l \}$ be such that each $P_i$ is a subset of $G^{k_i}$ where $k_i$ is called the {\em arity} of $P_i$. Note that unlike the usual definition of CSPs, we do not allow {\em shifts}, namely: for $b_1, \dots , b_{k_i} \in G$, $P_i + (b_1, \dots , b_{k_i})$ is not necessarily in $\calp$. Furthermore, predicates are allowed to have different arities. Let $k_{max} := \max_i k_i$ and $k_{min} := \min_i k_i$. An instance of CSP($\calp$) is denoted by $(V, \calc)$ where $V$ is a set of $n$ variables taking values in $G$. $\calc = \{ C_1, \dots , C_m \}$ is a set of $m$ constraints such that each $C_i$ is defined by its {\em type} $t_i \in \{ 1, 2, \dots , l \}$ (which represents the predicate corresponding to this constraint) and a tuple of $k_{t_i}$ variables $E_i = (e_{i, 1}, \dots , e_{i, k_{t_i}}) \in V^{k_{t_i}}$. 
In all instances in the paper, each variable appears at most once in each constraint. 
We sometimes abuse notation and regard $E_i$ as a subset of $V$ with cardinality $k_{t_i}$. We say that $(V, \calc)$ is $(s, \alpha)$-{\em expanding} if for any set of $s' \leq s$ constraints $\{ C_{i_1}, \dots , C_{i_{s'}} \} \subseteq \calc$, $|\bigcup_{1 \leq j \leq s'} E_{i_j}| \geq (\sum_{1 \leq j \leq s'} |E_{i_j}|) - \alpha \cdot s'$.
It is said to be  $(s, \alpha)$-{\em boundary expanding} if for any set of $s' \leq s$ constraints $\{ C_{i_1}, \dots , C_{i_{s'}} \} \subseteq \calc$, the number of variables appearing in exactly one constraint is at least $(\sum_{1 \leq j \leq s'} |E_{i_j}|) - \alpha \cdot s'$. 
Note that in both definitions, a smaller value of $\alpha$ corresponds to a better expansion. 
It is easy to see that $(s, \alpha)$-expansion implies $(s, 2\alpha)$-boundary expansion. An assignment $f : V \rightarrow G$ satisfies constraint $C_i$ if and only if $(f(e_{i, 1}), \dots , f(e_{i, k_{t_i}})) \in P_{t_i}$. 
When $G = \bool$, any instance of CSP($\calp$) is an instance of Min-Ones($\calp$), where the goal is to find an assignment $f$ that satisfies every constraint and minimizes $|f^{-1}(1)|$. 
The $k$-Uniform Hypergraph Vertex Cover ($k$-HVC) problem corresponds to Min-Ones($\{ \por \}$) where $\por(x_1, \dots , x_k) = 1$ if and only if there is at least one $1 \leq i \leq k$ with $x_i = 1$. 

\paragraph{Balanced Pairwise Independent Subsets and Distributions.}
Let $G$ be a finite set with $|G| = q$ and $k$ be a positive integer. Let $P$ be a subset of $G^k$ and $\mu$ be a distribution supported on $P$. The distribution
$\mu$ is said to be {\em balanced} if for all $i = 1, 2, \dots , k$ and $g \in G$, $\Pr_{(x_1, \dots , x_k) \sim \mu} [x_i = g] = \frac{1}{q}$.
It is called {\em balanced pairwise independent} if for all $i \neq j$ and $g, g' \in G$, 
$
\Pr_{(x_1, \dots , x_k) \sim \mu} [x_i = g \mbox{ and } x_j = g' ] = \frac{1}{q^2}.$
The predicate $P$ is called balanced (resp. balanced pairwise independent) if the uniform distribution on $P$ induces a balanced (resp. balanced pairwise independent) distribution on $P^k$.

\paragraph{Nearest Codeword.}
Fix the domain to be $\{ 0, 1 \}$. The Nearest Codeword problem is defined as Min-Ones($\{ P_{odd}, P_{even} \}$), where $x = (x_1, \dots , x_k) \in \{0, 1 \}^k$ belongs to $P_{odd}$ (resp. $P_{even}$) if and only if $|\{ i \in [k]: x_i = 1 \}|$ is an odd (resp. even) integer. We slightly abuse the notation and let $P_{odd}$ (resp. $P_{even}$) represent the odd (resp. even) predicates for all values of $k$. 
Let $B = (L \cup R, E_B)$ be the parity-check graph of some binary linear code with $|L| = n$ and $|R| = m$. Let $s \in \{0, 1\}^n$ be the received vector (i.e., the codeword which is corrupted by the noisy channel). Denote $R := \{1, \dots , m \}$. The instance of the Nearest Codeword problem given $s$ is given by $V = L$ and for each $1 \leq i \leq m$, $E_i = \{ v \in L : (v, i) \in E_B \}$, and $t_i = odd$ if $\sum_{v : (v, i) \in E_B} s_v = 1$ (summation over $\F_2$) and $t_i = even$ otherwise. 
In an integral assignment $f : L \rightarrow \{0, 1\}$, $f(v) = 1$ means that the $v$-th bit is flipped. So if all the constraints are satisfied, $(s_v + f(v))_{v \in L}$ is a valid codeword and $|f^{-1}(1)|$ is its Hamming distance to $s$. We say that $B$ is $(s, \alpha)$-expanding or $(s, \alpha)$-boundary expanding if the corresponding Nearest Codeword instance is so.

\paragraph{Sherali-Adams Hierarchy.}
Given an instance $(V, \calc)$ of CSP($\calp$) and a positive integer $t \leq |V|$, we define a {\em $t$-local distribution} to be a collection $\{ X_S(\alpha) \in [0, 1] \}_{S \subseteq V, |S| \leq t, \alpha : S \rightarrow G}$ satisfying $X_{\emptyset} = 1$ and for any $S \subseteq T \subseteq V$ with $|T| \leq t$ and for any $\alpha : S \rightarrow G$
\[
\sum_{\beta : T \setminus S \rightarrow G} X_{T} (\alpha \circ \beta) = X_S (\alpha),
\]
where $\alpha \circ \beta$ denotes an assignment $T \rightarrow G$ whose projections on $S$ and $T \setminus S$ are $\alpha$ and $\beta$ respectively. Given $t \geq k_{max}$, a solution to the $t$ rounds of the Sherali-Adams hierarchy is a $t$-local distribution. It is said to {\em satisfy} a constraint $C_i$ if for any $\alpha : E_i \rightarrow G$, 
$(\alpha(e_{i, 1}), \dots , \alpha(e_{i, k_{t_i}})) \notin P_i$ implies that $X_{E_i}(\alpha) = 0$ (i.e., the local distribution is only supported on the satisfying partial assignments). 
The solution is {\em balanced} if for any $v \in V$ and $g \in G$, $X_{v}(g) := X_{ \{ v \} } (v \mapsto g) = \frac{1}{|G|}$. If $G = \bool$, we say that the solution is {\em $p$-biased} if for any $v \in V$, $X_v(1) = p$. 

Given an LDPC code, let $d_c^{max}$ be the largest degree of any check node. 
The following claim, proved in Appendix~\ref{sec:dcmax_at_least_as_strong}, shows that a small number of rounds of the Sherali-Adams hierarchy is at least as strong as the basic LP of Feldman et al. 

\begin{claim}
\label{claim:sa_feldman}
The LP corresponding to $d_c^{max}$ rounds of the Sherali-Adams hierarchy is at least as strong as the LP of Feldman et al.
\end{claim}

\paragraph{Lasserre Hierarchy.}
Given an instance $(V, \calc)$ of CSP($\calp$) and an integer $t \leq |V|$, 
a solution to the $t$ rounds of the Lasserre hierarchy is a set of vectors 
$\{ V_S(\alpha) \}_{S \subseteq V, |S| \leq t, \alpha : S \rightarrow G}$, 
such that there exists a $2t$-local distribution $\{ X_S(\alpha) \}$ with the property: for any $S, T \subseteq V$ with $|S|, |T| \leq t$ and any $\alpha : S \rightarrow G$ and $\beta : T \rightarrow G$, we have that
\[
\langle V_S(\alpha), V_T(\alpha) \rangle = X_{S \cup T}(\alpha \circ \beta), 
\]
if $\alpha$ and $\beta$ are consistent on $S \cap T$, and $\langle V_S(\alpha), V_T(\alpha) \rangle = 0$ otherwise. The solution satisfies a constraint or is balanced if the corresponding local distribution is so.

\section{Solutions from Desired Structures}\label{sec:solns_des_struct}
In this section, we show how to construct solutions to the Sherali-Adams / Lasserre hierarchy for Min-Ones($\calp$) from desired structures. Given an instance of Min-Ones($\calp$) where $\calp = \{ P_1, \dots , P_l \}$ is a collection of predicates with $P_i \subseteq \bool^{k_i}$, we want to construct a solution to the Sherali-Adams / Lasserre hierarchy with small bias. However, in order to obtain a solution to the Sherali-Adams / Lasserre hierarchy for general CSPs, most current techniques~\cite{Schoenebeck08, GMT09, Tulsiani09, Chan13}  need a balanced pairwise independent distribution, and the resulting solution is typically {\em balanced} as well. Since the domain $G$ is fixed to $\bool$, a $\frac{1}{2}$-biased solution seems to be the best we can hope for; in fact, this is what Schoenebeck~\cite{Schoenebeck08} does for $k$-Hypergraph Vertex Cover in the Lasserre hierarchy thereby proving a gap of $2$ (for any $k \geq 3$). 

To bypass this barrier, we introduce the technique of {\em stretching} and {\em collapsing} the domain. Let $G'$ be a new domain with $|G'| = q$ and fix a mapping $\phi : G'  \rightarrow \bool$ (in every stretching in this paper, $|\phi^{-1}(1)| = 1$).
For each predicate $P_i$, let $P'_i$ be the corresponding new predicate $P'_i := \{ (g_1, \dots , g_{k_i}) \in (G')^{k_i} : (\phi(g_1), \dots , \phi(g_{k_i})) \in P_i \}$. 
Let $\calp' = \{ P'_1, \dots , P'_l \}$. Any instance $(V, \calc)$ of Min-Ones($\calp$) can be transformed to the instance $(V, \calc')$ of CSP($\calp'$) where variables in $V$ can take a value from $G'$ and each predicate $P_i$ is replaced by the predicate $P'_i$. The next lemma shows that any solution to the Sherali-Adams / Lasserre hierarchy for the new instance can be transformed to a solution for the old instance by {\em collapsing back} the domain. For $\beta : S \rightarrow \bool$, let $\phi^{-1}(\beta)$ be $\{\alpha : S \rightarrow G', \phi(\alpha(v)) = \beta(v) \mbox{ for all } v \in S \}$. 

\begin{lem}
\label{lem:collapsing}
Suppose that 
$\{ X'_S(\alpha) \}$ (resp. $\{ V'_S(\alpha) \}$) is a solution to the LP (resp. SDP) corresponding to $t$ rounds of the Sherali-Adams (resp. Lasserre) hiearchy for $(V, \calc')$ and that satisfies every constraint. Then, $\{ X_S(\beta) \}_{|S| \leq t, \beta : S \rightarrow \bool}$ 
(resp. $\{ V_S(\beta) \}_{|S| \leq t, \beta : S \rightarrow \bool}$) defined by
\[
X_S(\beta) = \sum_{\alpha \in \phi^{-1}(\beta)}X'_S(\alpha) \qquad 
(\mbox{resp. }V_S(\beta) = \sum_{\alpha \in \phi^{-1}(\beta)}V'_S(\alpha))
\]
is a valid solution to the $t$ rounds of the Sherali-Adams (resp. Lasserre) hiearchy for $(V, \calc)$ that satisfies every constraint. Furthermore, if the solution to the new instance is balanced, the obtained solution to the old instance is $\frac{1}{q}$-biased. 
\end{lem}
\begin{proof}
First, we prove the statment for the Sherali-Adams hierarchy. 
\paragraph{Sherali-Adams.}
By definition, we have that $X_{\emptyset} = X'_{\emptyset} = 1$, and $X_S(\alpha) \geq 0$. Moreover, for any $S \subseteq T \subseteq V$ with $|T| \leq t$ and for any $\beta : S \rightarrow \bool$, we have that
\begin{align*}
\sum_{\gamma : T \setminus S \rightarrow \bool} X_T(\beta \circ \gamma) &= 
\sum_{\gamma : T \setminus S \rightarrow \bool}
\sum_{\alpha \in \phi^{-1}(\beta \circ \gamma)} X'_T(\alpha) 
= \sum_{\beta' \in \phi^{-1}(\beta)}
\sum_{\gamma : T \setminus S \rightarrow \bool} 
\sum_{\gamma' \in \phi^{-1}(\gamma)} X'_T(\beta' \circ \gamma') \\
&= \sum_{\beta' \in \phi^{-1}(\beta)}
\sum_{\gamma : T \setminus S \rightarrow G'} X'_T(\beta' \circ \gamma') 
= \sum_{\beta' \in \phi^{-1}(\beta)}
X'_S(\beta') 
= X_S(\beta).
\end{align*}
Furthermore, if $\{ X'_S(\alpha) \}$ is balanced, then for any $v$, $
X_v(1) = \sum_{g \in \phi^{-1}(1)}X'_v(g) = \frac{|\phi^{-1}(1)|}{q} = \frac{1}{q}.$
This concludes the proof for the Sherali-Adams hierarchy.

\paragraph{Lasserre.}
Given a solution 
$\{ V'_S(\alpha) \}_{|S| \leq t, \alpha : S \rightarrow G}$ to the $t$ rounds of the Lasserre hierarchy, let $\{ X'_S(\alpha) \}_{|S| \leq 2t, \alpha : S \rightarrow G}$ be the $2t$-local distribution associated with $\{V'_S(\alpha) \}$. Let the $2t$-local distribution $\{ X_S(\beta) \}_{|S| \leq 2t, \beta : S \rightarrow \bool}$ be obtained from $\{ X'_S(\alpha) \}$ as as done above for the Sherali-Adams hierarchy. It is a valid $2t$-local distribution. We claim that $\{ X_S(\beta) \}$ is the local distribution associated with $\{ V_S(\beta) \}$. Fix $S, T$ such that $|S|, |T| \leq t$, $\beta : S \rightarrow \bool$ and $\gamma : T \rightarrow \bool$.
By the definition of $V_S(\beta)$ and $V_T(\gamma)$,
\[
 \langle V_S(\beta), V_T(\gamma) \rangle = 
\langle \sum_{\beta' \in \phi^{-1}(\beta)}V'_S(\beta'), \sum_{\gamma' \in \phi^{-1}(\gamma)}V'_T(\gamma') \rangle
= \sum_{\beta' \in \phi^{-1}(\beta)}\sum_{\gamma' \in \phi^{-1}(\gamma)}
\langle V'_S(\beta'), V'_T(\gamma') \rangle.
\]
If $\beta$ and $\gamma$ are inconsistent, then any $\beta' \in \phi^{-1}(\beta)$ and $\gamma' \in \phi^{-1}(\gamma)$ are inconsistent, and hence the RHS is 0 as desired. 
If they are consistent, then the RHS is equal to 
\[
\sum_{\beta' \in \phi^{-1}(\beta), \gamma'  \in \phi^{-1}(\gamma) \mbox{ consistent}}
\langle V'_S(\beta'), V'_T(\gamma') \rangle
= \sum_{\alpha' \in \phi^{-1}(\beta \circ \gamma)} X'_{S \cup T}(\alpha')
= X_{S \cup T}(\beta \circ \gamma).
\]
If $\{ V'_S(\alpha) \}$ is balanced, by definition $\{ X'_S(\alpha) \}$ is balanced, so the same proof for the Sherali-Adams hierarchy shows that $\{ V_S(\alpha) \}$ and $\{ X_S(\alpha)\}$ are $\frac{1}{q}$-biased.
\end{proof}


By Lemma~\ref{lem:collapsing} above, it suffices to construct a solution to the stretched instance. 
Theorems~\ref{thm:sa} and~\ref{thm:lasserre} below show that if the predicates $P_1, \dots, P_l$ satisfy certain desired properties and the instance is sufficiently expanding, there exists a balanced solution to the Sherali-Adams / Lasserre hierarchy. 
The proof is close to~\cite{GMT09} for the Sherali-Adams hierarchy and to~\cite{Schoenebeck08,Tulsiani09,Chan13} for the Lasserre hierarchy. 
Compared to their proofs for Max-CSPs, we have to deal with $2$ more issues. The first is that unlike usual CSPs, our definition of Min-Ones($\calp$) allows to use more than one predicate, and predicates can have different arities. The second is that for our purposes, the solution needs to be balanced (i.e., $X_v(g) = \frac{1}{|G|}$ for all $v,g$). We handle those differences by natural extensions of their techniques. The proofs are in Appendix~\ref{sec:proofs}. 
\begin{theorem}
Let $G$ be a finite set, $k_{min} \geq 3$, and $\calp = \{P_1, \dots , P_l \}$ be a collection of predicates such that each $P_i \subseteq G^{k_i}$ supports a balanced pairwise independent distribution $\mu_i$. 
Let $(V, \calc)$ be an instance of CSP($\calp$) such that $\calc$ is $(s, 2 + \delta)$-boundary expanding for some $0<\delta \leq \frac{1}{4}$. 
Then, there exists a balanced solution to the $\frac{\delta s}{6k_{max}}$ rounds of the Sherali-Adams hierarchy that satisfies every constraint in $\calc$. 
\label{thm:sa}
\end{theorem}

We point out that the updated version~\cite{BGMT12} of~\cite{GMT09} shows that their construction also works in the Sherali-Adams SDP hierarchy which is stronger than the original Sherali-Adams hierarchy but weaker than Lasserre. Both Theorems~\ref{thm:sa} and~\ref{thm:sa_ldpc} hold for the Sherali-Adams SDP hierarchy as well. In the proofs of Theorems~\ref{thm:sa} and~\ref{thm:sa_ldpc}, we focus on the original Sherali-Adams hierarchy to make the presentations simple.



\begin{theorem}
Let $G$ be a finite abelian group, $k_{min} \geq 3$ and $\calp = \{P_1, \dots , P_l \}$ be a collection of predicates such that each $P_i$ is a coset of a balanced pairwise independent subgroup of $G^{k_i}$. 
Let $(V, \calc)$ be an instance of CSP($\calp$) such that $\calc$ is $(s, 1 + \delta)$-expanding for $\delta \leq \frac{1}{4}$. 
Then, there exists a balanced solution to the $\frac{s}{16}$ rounds of the Lasserre hierarchy that satisfies every constraint in $\calc$. 
\label{thm:lasserre}
\end{theorem}

\section{Decoding Random $(d_v,d_c)$-LDPC Codes}
\label{sec:ldpc}
In this section, we apply Theorems~\ref{thm:sa} and~\ref{thm:lasserre} to random $(d_v,d_c)$-LDPC codes. In Section~\ref{subsec:sa}, we construct balanced pairwise independent distributions supported on even and odd predicates for different arity values and complete the proof of Theorem~\ref{thm:sa_ldpc} for Sherali-Adams. In Section~\ref{subsec:lasserre}, we show that both even and odd predicates contain cosets of balanced pairwise independent subgroups and introduce an additional technique based on taking the direct-sum of cosets of subgroups to conclude the proof of Theorem~\ref{thm:lasserre_ldpc} for Lasserre. We will need the next two lemmas which show that with high probability, a random $(d_v,d_c)$-LDPC code is almost regular and expanding. Their proofs use standard probabilistic arguments and appear in Appendix~\ref{sec:graphs}.

\begin{lem}\label{lem:mb_almost_regular}
Consider the parity-check graph of a random $(d_v, d_c)$-LDPC code. With high probability, every vertex on the left (resp. right) will have degree either $d_v$ or $d_v - 2$ (resp. $d_c$ or $d_c - 2$).
\end{lem}

\begin{lem}\label{lem:mb_expansion}
Given any $0 < \delta < 1/2$, there exists $\eta > 0$ (depending on $d_c$) such that the parity-check graph of a random $(d_v,d_c)$-LDPC code is $(\eta n, 1 + \delta)$-expanding with high probability.
\end{lem}

\subsection{Distributions for Sherali-Adams}
\label{subsec:sa}
To construct a solution for the Sherali-Adams hierarchy using Theorem~\ref{thm:sa}, we need each $P'_i \subseteq (G')^{k_i}$ to support a balanced pairwise independent distribution. For any $q \geq 2$ and $k = q + 1$, let $G' := \{ 0, 1, \dots , q - 1 \}$ and $\phi : G' \rightarrow \bool$ be defined by $\phi(0) = 1$ and $\phi(g) = 0$ for every $g \neq 0$. The {\em odd} and {\em even} predicates $P'_{odd}$ and $P'_{even}$ are defined by: $y \in P'_{odd}$ (resp. $P'_{even}$) if and only if $|\{ i \in [k] : y_i = 0 \}|$ is an odd (resp. even) integer. The choice of $k = q+1$ is optimal since, as shown in Lemma \ref{le:imp_k_k_minus_1} in Appendix~\ref{sec:sa_more}, if $k = q$, there is no balanced pairwise independent distribution that is supported on the even larger predicate $\{y \in (G')^k : y_i = 0 \mbox{ for some } i \}$ which contains $P'_{odd}$. Set $p := 1/q$. To construct a distribution on $y \in (G')^k$, we will show how to sample $x \in \{0, 1 \}^k$. Given $x$, each $y_i$ is set to $0$ if $x_i = 0$ and uniformly sampled from $\{ 1, \dots, q - 1 \}$ otherwise. It is easy to see that when this distribution on $x$ is $(1 - p)$-biased (i.e. $\Pr[x_i = 0] = p$ for all $i$) and pairwise independent (i.e. $\Pr[x_i = x_j = 0] = p^2$) for all $i \neq j$), $y$ becomes balanced pairwise independent. Furthermore, $x$ and $y$ have the same number of 0's. Therefore, it suffices to show how to sample a $(1 - p)$-biased pairwise independent vector $x$. 

\paragraph{Odd predicate, Odd $k \geq 3$, $q = k - 1$.}
Let $\0 := (0, \dots, 0)$, $\1 = (1, \dots ,1)$ and $ e_i$ be the $i$-th unit vector. 
Sample $x \in (G')^k$ from the distribution with probability mass function: $\Pr[x = \0] = p^2$ and $\Pr[x = \1 - e_i] = \frac{1 - p^2}{k}$ for each $i$. Each support-vector has an odd number of 0's. For any $i$, $\Pr[x_i = 0] = \Pr[x = e_i] + \Pr[x = \1] = \frac{1 - p^2}{k} + p^2 = p.$ For any $i \neq j$, $\Pr[x_i = x_j = 0] = \Pr[x_i = 1] = p^2$. This simple construction is optimal: If $k = q + 1$ is even, Lemma~\ref{le:imp_k_q+1_even} (in Appendix~\ref{sec:sa_more}) shows that there is no such balanced pairwise independent distribution supported in $P'_{odd}$. 

\paragraph{Even Predicate, $k \geq 3$, $q = k - 1$.}
Sample $x \in (G')^k$ from the distribution with probability mass function: $\Pr[x = \1 - e_i - e_j] = p^2$ for each $i \neq j$ and $\Pr[x = \1] = 1 - p^2 \binom{k}{2} = \frac{1 - p}{2}$. Each support-vector has an even number of 0's. For any $i$, $\Pr[x_i = 0] = \Pr[\exists j \neq i: x = \1 - e_i - e_j] = p^2 (k - 1) = p$. For $i \neq j$, $\Pr[x_i = x_j = 0] = \Pr[x = \1 - e_i - e_j] = p^2$. 

\paragraph{Other values of $k$ and $q$.}
If $k \geq 4$ is an even integer, we show in Lemma \ref{le:imp_k_q+1_even} of Appendix~\ref{sec:sa_more} that for $q = k - 1$, there is no balanced pairwise independent distribution that is supported in the odd predicate. However, it is still possible to have such a distribution when $q = k - 2$ for both odd and even predicates. In Lemma~\ref{sa:predicates} below (whose proof appears in Appendix~\ref{sec:sa_more}), we prove the existence of pairwise independent distributions supported in the odd and even predicates for slightly smaller values of $q$ (in terms of $k$). These distributions will be used to handle instances where the constraints have different arities.

\begin{lem}
Let $G = \{ 0, 1, \dots , q - 1 \}$ be a finite set. For the following combinations of arity values $k$ and alphabet size values $q$, each of the odd predicate and the even predicate supports a balanced pairwise independent distribution on $G^k$: (i) Any even integer $k \geq 4$ with $q = k - 2$, (ii) Any odd integer $k \geq 5$ with $q = k - 3$ and (iii) Any even integer $k \geq 6$ with $q = k - 4$.
\label{sa:predicates}
\end{lem}

The constructed distributions for the Sherali-Adams hierarchy are summarized in Table~\ref{tab:SA_construct_table}.

\begin{table}
\begin{center}
    \begin{tabular}{|l|l|l|l|l|}
    \hline
    \multirow{2}{*}{\diagbox{Type}{Arity}} & \multicolumn{2}{|c|}{$d_c$ odd ($q=d_c-3$)} & \multicolumn{2}{|c|}{$d_c$ even ($q=d_c-4$)}\\ \cline{2-5}
     & $k=d_c$ & $k=d_c-2$ & $k = d_c$ & $k = d_c-2$\\ \hline
    Odd & Lemma~\ref{sa:predicates} $(ii)$ & Section~\ref{subsec:sa} & Lemma~\ref{sa:predicates} $(iii)$ & Lemma~\ref{sa:predicates} $(i)$ \\ \hline
    Even & Lemma~\ref{sa:predicates} $(ii)$ & Section~\ref{subsec:sa} & Lemma~\ref{sa:predicates} $(iii)$ & Lemma~\ref{sa:predicates} $(i)$ \\ 
    \hline
    \end{tabular}
\end{center}
\caption{Distributions for Sherali-Adams}
  \label{tab:SA_construct_table}
\end{table}

\paragraph{Proof of Theorem~\ref{thm:sa_ldpc}.}
Consider a random $(d_v, d_c)$-LDPC code and fix $\delta = 1/8$. 
Lemma~\ref{lem:mb_almost_regular} and Lemma~\ref{lem:mb_expansion} ensure that with high probability, the degree of each check node is either $d_c$ or $d_c-2$ and there exists $\eta > 0$ such that the code is $(\eta n, 1 + \delta)$-expanding, and hence $(\eta n, 2 + 2\delta)$-boundary expanding. 
For any received vector, let $(V, \calc)$ be the corresponding instance of Nearest Codeword. 
Let $q = d_c - 3$ (resp. $d_c - 4$) if $d_c$ is odd (resp. even). 
Stretch the domain from $\{ 0, 1 \}$ to $G' := \{0, 1, \dots q - 1 \}$. The above constructions show that for any $k \in \{ d_c, d_c - 2\}$ and $type \in \{ even, odd \}$, $P_{ type } \subseteq (G')^k$ supports a balanced pairwise independent distribution. Theorem~\ref{thm:sa} gives a balanced solution to the $\frac{2\delta \eta n }{6 d_c} = \frac{\eta n}{24 d_c}$ rounds of the Sherali-Adams hierarchy that satisfies every constraint in the stretched instance. Lemma~\ref{lem:collapsing} transforms this solution to a $\frac{1}{q}$-biased solution to the same number of rounds for the original Nearest Codeword instance. 

\subsection{Subgroups for Lasserre}
\label{subsec:lasserre}
As in the Sherali-Adams hierarchy, to find a good solution in the Lasserre hierarchy, it suffices to construct a stretched instance. To construct a solution in the Lasserre hierarchy via Theorem~\ref{thm:lasserre}, we need the stretched domain $G'$ to be a finite {\em abelian group } and each stretched predicate $P'_i$ to be a coset of a balanced pairwise independent {\em subgroup} of $(G')^k$. We will first construct such predicates for $q$ being any power of $2$ and $k = q + 1$. For such $q$ and $k$, let $G' := \mathbb{F}_q$ and $\phi : G' \rightarrow \bool$ be defined by $\phi(0) = 1$ and $\phi(g) = 0$ for every $g \neq 0$. As for Sherali-Adams, the predicates $P'_{odd}$ and $P'_{even}$ are defined in the natural way, namely: $(x_1, \dots , x_k) \in P'_{odd}$ (resp. $P'_{even}$) if and only if $|\{ i \in [k] : x_i = 0 \}|$ is an odd (resp. even) integer. We show that each of $P'_{odd}$ and $P'_{even}$ contains a coset of a balanced pairwise independent subgroup of $(G')^k$. 

\paragraph{Odd Predicate, $k = 2^i + 1$, $q = k - 1$.}
For the odd predicate $P'_{odd}$, we actually show that it contains a balanced pairwise independent subgroup of $(G')^k$. 
Let $\{ \alpha x + \beta y \}_{\alpha, \beta \in \F_q}$ be the set of all $q^2$ bivariate linear functions over $\F_q$. Let $E := \{ (0, 1) \} \cup \{ (1, a) \}_{a \in \F_q}$ be the set of $q + 1 = k$ evaluation points. Our subgroup is defined by $H' := \{ (\alpha x + \beta y )_{(x, y) \in E} \}_{\alpha, \beta \in \mathbb{F}_q}$. Note that $H'$ is a subgroup of $(G')^k$. In general, there are $q+1$ distinct lines passing through the origin in the $\F_q^2$-plane; our set $E$ contains exactly one point from each of those lines. The balanced pairwise independence of $H'$ follows from Lemma~\ref{lem:pi}.

\begin{lem}
Let $d \in \mathbb{N}$ and $E \subseteq \F_q^d \setminus \{ 0 \}$ contain at most one point from each line passing the origin. Then, the subgroup $\{ ( \sum_{i=1}^d \alpha_i x_i )_{(x_1, \dots , x_d) \in E} \}_{\alpha_1, \dots , \alpha_d \in \mathbb{F}_q}$ is balanced pairwise independent. 
\label{lem:pi}
\end{lem}

\begin{proof}
Let $(b_1, \dots , b_d) \neq (c_1, \dots , c_d) \in E$ be two points not on the same line passing through the origin. For balanced pairwise independence, we need $(\sum_i \alpha_i b_i, \sum_i \alpha_i c_i)_{\alpha_1, \dots , \alpha_d \in \F_q}$ to be the uniform distribution on $\F_q^2$. Since there are exactly $q^d$ choices for the tuple $(\alpha_1,\dots,\alpha_d)$, for any $\beta, \gamma \in \F_q$, it suffices to show that there exists $q^{d - 2}$ choices of the tuple $(\alpha_1, \dots, \alpha_d) \in \F_q$ such that $\sum_i \alpha_i b_i = \beta, \sum_i \alpha_i c_i = \gamma$. Since the two points are not on the same line through the origin, there must be two indices $i \neq j$ such that $b_i c_j \neq b_j c_i$. Without loss of generality, assume that $i = 1$ and $j = 2$. For any choice of $(\alpha_3, \dots , \alpha_{d})$, there is exactly one solution $(\alpha_1, \alpha_2)$ to the system:
$$\alpha_1 b_1 + \alpha_2 b_2 = \beta - \sum_{i=3}^{d} \alpha_i b_i$$
$$\alpha_1 c_1 + \alpha_2 c_2 = \gamma - \sum_{i=3}^{d} \alpha_i c_i$$
\end{proof}

The next lemma concludes the analysis of the odd predicate.
\begin{lem}\label{lem:odd_num_zeros}
Each element of $H'$ has an odd number of $0$ coordinates.
\end{lem}

\begin{proof}
Recall that $k = q+1$ with $q$ a power of $2$ and $G' := \mathbb{F}_q$. Our set of evaluation points is defined by
$$E := \{ (0, 1) \} \cup \{ (1, a) \}_{a \in \F_q}$$
and our subgroup $H'$ of $(G')^k$ is defined by
$$H' := \{ (\alpha x + \beta y )_{(x, y) \in E} \}_{\alpha, \beta \in F_q}$$
Let $h_{\alpha, \beta} := (\alpha x + \beta y )_{(x, y) \in E}$ be any element of $H'$ (where $\alpha, \beta \in \mathbb{F}_q$). The fact that $h_{\alpha, \beta}$ has an odd number of $0$ coordinates can be seen by distinguishing the following three cases:
\begin{itemize}
\item For $\alpha = \beta = 0$: $h_{\alpha, \beta} = (0, 0, \dots , 0)$, which has $k$ $0$ coordinates, and $k$ is set to be an odd integer. 
\item For $\beta = 0$ and $\alpha \neq 0$: $(0, 1)$ is the unique zero of the function $\alpha x + \beta y$ in $E$. 
\item For $\beta \neq 0$: $(1, \alpha/\beta)$ is the unique zero of the function $\alpha x + \beta y$ in $E$. 
\end{itemize}
\end{proof}

\paragraph{Even Predicate, $k = 2^i + 1$, $q = k - 1$.}
Dealing with $P'_{even}$ is more difficult, since $P'_{even}$ will not contain any subgroup: this can be seen by observing that the zero element $(0, 0, \dots , 0) \in (G')^k$ has an odd number of $0$ coordinates and should be in any subgroup. Instead, we show that $P'_{even}$ will contain a \emph{coset} of a balanced pairwise independent subgroup. As in the above case of the odd predicate, our subgroup $H'$ will be of the form $\{ (\alpha x + \beta y )_{(x, y) \in E'} \}_{\alpha, \beta \in \mathbb{F}_q}$, for some subset $E' \subseteq \F_q^2$ of $q + 1 = k$ evaluation points. As before, the set $E'$ will contain exactly one non-zero point on each line passing through the origin and hence balanced pairwise independence will follow from Lemma~\ref{lem:pi}. Moreover, the set $E'$ will have the property that $H' - (1, 1, \dots , 1) \subseteq P'_{even}$; i.e, for any $\alpha, \beta \in \mathbb{F}_q$, there is an even number of points $(x, y) \in E'$ satisfying the equation $\alpha x + \beta y = 1$. For example, if $\alpha = \beta = 0$, no point satisfies this equation. If at least one of $\alpha,\beta$ is nonzero, then $\{ \alpha x + \beta y = 1 \}_{(\alpha, \beta) \in \F_q^2 \setminus \{(0,0)\}}$ consists of all ($q^2 - 1$) distinct lines not passing through the origin. Thus, we set $E':= E \setminus \{0\}$ where $E$ is the set which is guaranteed to exist by Lemma~\ref{le:zero_or_two}.

\begin{lemma}\label{le:zero_or_two}
For every $q$ that is a power of $2$, there is a subset $E \subseteq \mathbb{F}_q^2$ containing the origin $(0,0)$ such that $|E| = q+2$ and every line in the $\mathbb{F}_q^2$-plane contains either $0$ or $2$ points in $E$.
\end{lemma}

\begin{proof}
Consider the map $h: \mathbb{F}_q \to \mathbb{F}_q$ given by $h(a) = a^2 +a$. Since $h(a) = h(a+1)$ for all $a \in \mathbb{F}_q$, we can see that $h$ is two-to-one. Hence, there exists $\eta \in \mathbb{F}_q$ such that the polynomial $g(a) = a^2 + a + \eta$ has no roots in $\mathbb{F}_q$. Fix such an $\eta$. Define the map $f: \mathbb{F}_q \to \mathbb{F}_q$ by $f(a) = (g(a))^{-1}$ for all $a \in \mathbb{F}_q$. Note that since $g$ has no roots in $\mathbb{F}_q$, $f$ is well defined and non-zero on $\mathbb{F}_q$. Now let $E := \{(0,0)\} \cup \{(0,1)\} \cup \{f(a) (1,a): a \in \mathbb{F}_q\}$. We next argue that every line $l$ in $\mathbb{F}_q^2$ contains either $0$ or $2$ points in $E$. We distingish several cases:
\begin{itemize}
\itemsep0em 
\item $l$ contains the origin $(0,0)$: If $l$ is a vertical line, then it has the form $l: (x = 0)$ and $(0,1)$ is the only other point of $E$ that lies on $l$. Henceforth, assume that $l$ is non-vertical. Then, it has the form $l: (y = \alpha x)$ for some $\alpha \in \mathbb{F}_q$. In this case, the unique other point of $E$ that lies on $l$ is $f(\alpha) (1,\alpha)$.
\item $l$ doesn't contain $(0,0)$ but contains $(0,1)$: Thus, it is of the form $l: (y = \alpha x + 1)$ for some $\alpha \in \mathbb{F}_q$. Then, a point $f(a) (1,a)$ lies on $l$ if and only if $a f(a) = \alpha f(a) + 1$ which is equivalent to $a = \alpha + g(a)$. This means that $a$ is a root of the polynomial $g(a) + \alpha - a = a^2+\eta + \alpha$. By Lemma \ref{le:quadratic_num_roots} below, this polynomial has a unique root (of multiplicity $2$) in $\mathbb{F}_q$. So $l$ contains exactly $2$ points in $E$.
\item $l$ contains neither $(0,0)$ nor $(0,1)$: If $l$ is a vertical line, then it has the form $l: (x = \beta)$ for some $\beta \in \mathbb{F}_q \setminus \{0\}$. Then, a point $f(a) (1,a)$ lies on $l$ if and only if $f(a) = \beta$, which is equivalent to $g(a) = \beta^{-1}$ (since $\beta \neq 0$). This means that $a$ is a root of the polynomial $g(a) - \beta^{-1} = a^2 + a + \eta - \beta^{-1}$. By Lemma \ref{le:quadratic_num_roots} below, this polynomial has either $0$ or $2$ roots in $\mathbb{F}_q$. Hence, $l$ contains either $0$ or $2$ points in $E$. Henceforth, assume that $l$ is non-vertical. Then, it has the form $l: (y = \alpha x + \beta)$ for some $\alpha \in \mathbb{F}_q$ and $\beta \in \mathbb{F}_q \setminus \{0,1\}$. Then, a point $f(a) (1,a)$ lies on $l$ if and only if $a f(a) = \alpha f(a) + \beta$, which is equivalent to $a = \alpha + \beta g(a)$. This is equivalent to $g(a) = a/\beta - \alpha/\beta$. This means that $a$ is a root of the polynomial $g(a) - a/\beta + \alpha/\beta = a^2 + a (1-1/\beta) + \eta + \alpha/\beta$. By Lemma \ref{le:quadratic_num_roots} below and since $\beta \neq 1$, this polynomial has either $0$ or $2$ roots in $\mathbb{F}_q$. So $l$ contains either $0$ or $2$ points in $E$.$\qedhere$
\end{itemize}
\end{proof}

\begin{lem}\label{le:quadratic_num_roots}
Let $q$ be a power of $2$. Then, a quadratic polynomial $p(a) = a^2 + c_1 a + c_0$ over $\mathbb{F}_q$ has a unique root (of multiplicity $2$) if and only if $c_1 = 0$.
\end{lem}

\begin{proof}
If $p(a)$ has a unique root $\lambda \in \mathbb{F}_q$, then $(a - \lambda)$ divides $p(a)$ and hence $p(a) = (a - \lambda)^2 = a^2 - 2 \lambda a + \lambda^2$. Since $\mathbb{F}_q$ has characteristic $2$, we get that $p(a) = a^2 + \lambda^2$ and we conclude that $c_1 = 0$. Conversely, assume that $p(a) = a^2 + c_0$ for some $c_0 \in \mathbb{F}_q$. Since $\mathbb{F}_q$ has characteristic $2$, the map $\kappa: \mathbb{F}_q \to \mathbb{F}_q$ given by $\kappa(a) = a^2$ is a bijection. Hence, there exists $\lambda \in \mathbb{F}_q$ such that $\kappa(\lambda) = \lambda^2 = c_0$. Using again the fact that $\mathbb{F}_q$ has characteristic $2$ , we conclude that $p(a) = a^2 - \lambda^2 = (a - \lambda)^2$ and hence $p(a)$ has a unique root (of multiplicity $2$) in $\mathbb{F}_q$.
\end{proof}

\paragraph{Even Predicate, $q = 2^i$, $k = 2q$.}
Since a check node in a random $(d_v,d_c)$-LDPC code has degree $d_c$ or $d_c-2$, we need to construct even and odd predicates for both arities $d_c$ and $d_c-2$ and over the same alphabet. We first construct an additional even predicate with arity $k = 2q$ based on \emph{trivariate} linear forms.

\begin{lem}
Let $q$ be a power of $2$ and $k = 2q$. There exists a subgroup of $\F_q^k$ such that every element in the subgroup contains an even number of $0$ coordinates. 
\label{lem:lasserre_2q}
\end{lem}

\begin{proof}
See Appendix~\ref{sec:subgroup_more}.
\end{proof}

\paragraph{Direct sums of cosets of subgroups}
For any $q = 2^i$, we constructed $3$ cosets of subgroups:
$H_1 \subseteq \F_q^{q+1}$ contained in the odd predicate, 
$H_2 \subseteq \F_q^{q+1}$ contained in the even predicate, 
$H_3 \subseteq \F_q^{2q}$ contained in the even predicate. Any direct sum of them gives a coset of a subgroup of $\F_q^k$ with $k$ being the sum of the individual arities. If we add one coset contained in the even predicate and one contained in the odd predicate, the direct sum will be contained in the odd predicate. On the other hand, if we add two cosets that are contained in the same (even or odd) predicate, the direct sum will be contained in the even predicate. For $d_c = 3q + 3$, we use such direct sums to construct the desired even and odd predicates for arities $d_c$ and $d_c - 2$ as follows:
\begin{itemize}
\itemsep0em 
\item $H_1 \oplus H_1 \oplus H_1$: A coset of a subgroup of $\F_q^{3q + 3}$, contained in the odd predicate. 
\item $H_1 \oplus H_1 \oplus H_2$: A coset of a subgroup of $\F_q^{3q + 3}$, contained in the even predicate. 
\item $H_1 \oplus H_3$: A coset of a subgroup of $\F_q^{3q + 1}$, contained in the odd predicate. 
\item $H_2 \oplus H_3$: A coset of a subgroup of $\F_q^{3q + 1}$, contained in the even predicate. 
\end{itemize}


The constructed subgroups for the Lasserre hierarchy are summarized in Table~\ref{tab:Lasserre_construct_table}.

\begin{table}
\begin{center}
    \begin{tabular}{|l|l|l|l|l|}
    \hline
    \diagbox{Type}{Arity} & $q+1$ & $2q$ & $d_c-2 = 3q+1$ & $d_c = 3q+3$ \\ \hline
    Odd & Lemma~\ref{lem:odd_num_zeros} $(H_1)$ &  & $H_1 \oplus H_3$ & $H_1 \oplus H_1 \oplus H_1$\\ \hline
    Even & Lemma~\ref{le:zero_or_two} $(H_2)$ & Lemma~\ref{lem:lasserre_2q} $(H_3)$ & $H_2 \oplus H_3$ & $H_1 \oplus H_1 \oplus H_2$ \\ 
    \hline
    \end{tabular}
\end{center}
\caption{Subgroups for Lasserre}
  \label{tab:Lasserre_construct_table}
\end{table}

\paragraph{Proof of Theorem~\ref{thm:lasserre_ldpc}.}
Consider a random $(d_v, d_c)$-LDPC code when $d_c = 3 \cdot 2^i + 3$ and fix $\delta = 1/8$, $q = 2^i = \frac{d_c - 3}{3}$. 
Lemmas~\ref{lem:mb_almost_regular} and~\ref{lem:mb_expansion} ensure that with high probability, each check-degree is either $d_c$ or $d_c-2$ and the code is $(\eta n, 1 + \delta)$-expanding for some $\eta > 0$. For any received vector, let $(V, \calc)$ be the corresponding instance of Nearest Codeword. Stretch the domain from $\{ 0, 1 \}$ to $G' := \F_q$. The above constructions show that for any $k \in \{ d_c, d_c - 2\}$ and $type \in \{ even, odd \}$, $P_{ type } \subseteq (G')^k$ is a coset of a balanced pairwise independent subgroup. Theorem~\ref{thm:lasserre} gives a balanced solution to the $\frac{\eta n }{16}$ rounds of the Lasserre hierarchy that satisfies every constraint in the stretched instance. Lemma~\ref{lem:collapsing} transforms this solution to a $\frac{1}{q}$-biased solution to the same number of rounds for the original Nearest Codeword instance.

\section{Conclusion}

In this work, we showed that fairly powerful extensions of LP decoding, based on the Sherali-Adams and Lasserre hierarchies, fail to correct much more errors than the basic LP-decoder.
It would be interesting to extend our Lasserre lower bounds for all values of $d_c$, which seems to require some new technical ideas. Finally, it would be very interesting to understand whether LP/SDP hierarchies can come close to capacity on \emph{irregular} ensembles \cite{richardson2001design} or on \emph{spatially-coupled} codes \cite{kudekar2012spatially}.

\section*{Acknowledgments}
The authors would like to thank Madhu Sudan, Venkatesan Guruswami, Mohammad Bavarian, Louay Bazzi, David Burshtein, Siu On Chan, R\"udiger Urbanke, Ameya Velingker and the anonymous reviewers for very helpful discussions and pointers.

\bibliographystyle{alpha}
\bibliography{coding_gaps}

\newcommand{\etalchar}[1]{$^{#1}$}
\begin{thebibliography}{DDKW08}

\bibitem[ACG{\etalchar{+}}10]{ACGKTZ10}
Matthew Andrews, Julia Chuzhoy, Venkatesan Guruswami, Sanjeev Khanna, Kunal
  Talwar, and Lisa Zhang.
\newblock Inapproximability of edge-disjoint paths and low congestion routing
  on undirected graphs.
\newblock {\em Combinatorica}, 30(5):485--520, 2010.

\bibitem[ADS12]{arora2012message}
Sanjeev Arora, Constantinos Daskalakis, and David Steurer.
\newblock Message-passing algorithms and improved {LP} decoding.
\newblock {\em IEEE Transactions on Information Theory}, 58(12):7260--7271,
  2012.

\bibitem[BGMT12]{BGMT12}
Siavosh Benabbas, Konstantinos Georgiou, Avner Magen, and Madhur Tulsiani.
\newblock {SDP} gaps from pairwise independence.
\newblock {\em Theory of Computing}, 8(1):269--289, 2012.

\bibitem[BM02]{burshtein2002bounds}
David Burshtein and Gadi Miller.
\newblock Bounds on the performance of belief propagation decoding.
\newblock {\em IEEE Transactions on Information Theory}, 48(1):112--122, 2002.

\bibitem[BM04]{burshtein2004asymptotic}
David Burshtein and Gadi Miller.
\newblock Asymptotic enumeration methods for analyzing ldpc codes.
\newblock {\em IEEE Transactions on Information Theory}, 50(6):1115--1131,
  2004.

\bibitem[Cha13]{Chan13}
Siu~On Chan.
\newblock Approximation resistance from pairwise independent subgroups.
\newblock In {\em Proceedings of the 45th annual ACM Symposium on Theory of
  Computing}, STOC '13, pages 447--456, 2013.

\bibitem[DDKW08]{daskalakis2008probabilistic}
Constantinos Daskalakis, Alexandros~G Dimakis, Richard~M Karp, and Martin~J
  Wainwright.
\newblock Probabilistic analysis of linear programming decoding.
\newblock {\em IEEE Transactions on Information Theory}, 54(8):3565--3578,
  2008.

\bibitem[DGKR05]{DGKR05}
Irit Dinur, Venkatesan Guruswami, Subhash Khot, and Oded Regev.
\newblock A new multilayered {PCP} and the hardness of hypergraph vertex cover.
\newblock {\em SIAM Journal on Computing}, 34(5):1129--1146, 2005.

\bibitem[DPT{\etalchar{+}}02]{di2002finite}
Changyan Di, David Proietti, I.~Emre Telatar, Thomas~J Richardson, and
  R{\"u}diger~L Urbanke.
\newblock Finite-length analysis of low-density parity-check codes on the
  binary erasure channel.
\newblock {\em IEEE Transactions on Information Theory}, 48(6):1570--1579,
  2002.

\bibitem[Fel03]{feldman2003decoding}
Jon Feldman.
\newblock {\em Decoding error-correcting codes via linear programming}.
\newblock PhD thesis, Massachusetts Institute of Technology, 2003.

\bibitem[FMS{\etalchar{+}}07]{feldman2007lp}
Jon Feldman, Tal Malkin, Rocco~A Servedio, Clifford Stein, and Martin~J
  Wainwright.
\newblock {LP} decoding corrects a constant fraction of errors.
\newblock {\em IEEE Transactions on Information Theory}, 53(1):82--89, 2007.

\bibitem[FWK05]{feldman2005using}
Jon Feldman, Martin~J Wainwright, and David~R Karger.
\newblock Using linear programming to decode binary linear codes.
\newblock {\em IEEE Transactions on Information Theory}, 51(3):954--972, 2005.

\bibitem[Gal62]{gallager1962low}
Robert~G Gallager.
\newblock Low-density parity-check codes.
\newblock {\em IRE Transactions on Information Theory}, 8(1):21--28, 1962.

\bibitem[GMT09]{GMT09}
Konstantinos Georgiou, Avner Magen, and Madhur Tulsiani.
\newblock Optimal {S}herali-{A}dams gaps from pairwise independence.
\newblock In {\em Approximation, Randomization, and Combinatorial Optimization.
  Algorithms and Techniques}, volume 5687 of {\em Lecture Notes in Computer
  Science}, pages 125--139. 2009.

\bibitem[KRU12]{kudekar2012spatially}
Shrinivas Kudekar, Tom Richardson, and R{\"u}diger Urbanke.
\newblock Spatially coupled ensembles universally achieve capacity under belief
  propagation.
\newblock In {\em Proceedings of 2012 IEEE International Symposium on
  Information Theory}, ISIT 2012, pages 453--457, 2012.

\bibitem[KSTW01]{KSTW01}
Sanjeev Khanna, Madhu Sudan, Luca Trevisan, and David~P Williamson.
\newblock The approximability of constraint satisfaction problems.
\newblock {\em SIAM Journal on Computing}, 30(6):1863--1920, 2001.

\bibitem[Las01]{Lasserre01}
Jean~B. Lasserre.
\newblock An explicit exact sdp relaxation for nonlinear 0-1 programs.
\newblock In {\em Proceedings of the 8th conference on Integer Programming and
  Combinatorial Optimization}, IPCO '01, pages 293--303, 2001.

\bibitem[Lau03]{laurent2003comparison}
Monique Laurent.
\newblock A comparison of the {S}herali-{A}dams, {L}ov{\'a}sz-{S}chrijver, and
  {L}asserre relaxations for 0-1 programming.
\newblock {\em Mathematics of Operations Research}, 28(3):470--496, 2003.

\bibitem[LS02]{litsyn2002ensembles}
Simon Litsyn and Vladimir Shevelev.
\newblock On ensembles of low-density parity-check codes: asymptotic distance
  distributions.
\newblock {\em IEEE Transactions on Information Theory}, 48(4):887--908, 2002.

\bibitem[Mac99]{mackay1999good}
David~JC MacKay.
\newblock Good error-correcting codes based on very sparse matrices.
\newblock {\em IEEE Transactions on Information Theory}, 45(2):399--431, 1999.

\bibitem[MB01]{MB01}
G.~Miller and D.~Burshtein.
\newblock Bounds on the maximum-likelihood decoding error probability of
  low-density parity-check codes.
\newblock {\em IEEE Transactions on Information Theory}, 47(7):2696--2710, Nov
  2001.

\bibitem[RSU01]{richardson2001design}
Thomas~J Richardson, Mohammad~Amin Shokrollahi, and R{\"u}diger~L Urbanke.
\newblock Design of capacity-approaching irregular low-density parity-check
  codes.
\newblock {\em IEEE Transactions on Information Theory}, 47(2):619--637, 2001.

\bibitem[RU01]{richardson2001capacity}
Thomas~J Richardson and R{\"u}diger~L Urbanke.
\newblock The capacity of low-density parity-check codes under message-passing
  decoding.
\newblock {\em IEEE Transactions on Information Theory}, 47(2):599--618, 2001.

\bibitem[RU08]{richardson2008modern}
Tom Richardson and Ruediger Urbanke.
\newblock {\em Modern coding theory}.
\newblock Cambridge University Press, 2008.

\bibitem[SA90]{SA90}
H.~Sherali and W.~Adams.
\newblock A hierarchy of relaxations between the continuous and convex hull
  representations for zero-one programming problems.
\newblock {\em SIAM Journal on Discrete Mathematics}, 3(3):411--430, 1990.

\bibitem[Sch08]{Schoenebeck08}
G.~Schoenebeck.
\newblock Linear level {L}asserre lower bounds for certain k-{CSP}s.
\newblock In {\em Proceedings of the 49th annual IEEE symposium on Foundations
  of Computer Science}, FOCS '08, pages 593--602, Oct 2008.
\newblock Newer version available at the author's homepage.

\bibitem[SS94]{sipser1994expander}
Michael Sipser and Daniel~A Spielman.
\newblock Expander codes.
\newblock In {\em Proceedings of the 54th annual IEEE symposium on Foundations
  of Computer Science}, FOCS 1994, pages 566--576, 1994.

\bibitem[Tul09]{Tulsiani09}
Madhur Tulsiani.
\newblock {CSP} gaps and reductions in the {L}asserre hierarchy.
\newblock In {\em Proceedings of the 41st annual ACM Symposium on Theory of
  Computing}, STOC '09, pages 303--312, 2009.

\bibitem[Vid13]{viderman2013lp}
Michael Viderman.
\newblock {LP} decoding of codes with expansion parameter above 2/3.
\newblock {\em Information Processing Letters}, 113(7):225--228, 2013.

\end{thebibliography}

\appendix

\section{LP Decoding and the Sherali-Adams Hierarchy}\label{sec:dcmax_at_least_as_strong}
Fix a code represented by its parity-check graph $G = ([n] \cup [m], E)$, and let $N(j)$ be the set of all neighbors of check node $j$. The LP relaxation of  Feldman et al. is given by:
$$\min \frac{1}{n} \displaystyle\sum\limits_{i=1}^n f_i$$
subject to:
$$\forall j \in [m], ~ \displaystyle\sum\limits_{S \in E_j} w_{j,S} = 1$$
$$\forall (i,j) \in E, ~ \displaystyle\sum\limits_{S \in E_j, S \ni i} w_{j,S} = f_i$$
$$ \forall i \in [n], ~ 0 \le f_i \le 1 $$
$$ \forall j \in [m], ~ \forall S \in E_j, ~ w_{j,S} \geq 0 $$
where $E_j$ is the set of all subsets of $N(j)$ of even (resp. odd) cardinality depending on whether the received vector has an even (resp. odd) number of $1$'s in $N(j)$.

\begin{claim}
[Restatement of Claim~\ref{claim:sa_feldman}]
The LP corresponding to $d_c^{max}$ rounds of the Sherali-Adams hierarchy is at least as strong as the above LP relaxation of Feldman et al..
\end{claim}

\begin{proof}
To prove this claim, it is enough to map any feasible solution to the LP corresponding to $d_c^{max}$ rounds of the Sherali-Adams hierarchy into a feasible solution to the LP of Feldman et al. with the same objective value. The map is the following:
\begin{itemize}
\item For every $i \in [n]$, let $f_i = X_{\{i\}}(1)$.
\item For every $j \in [m]$ and every $S \subseteq N(j)$, let $w_{j,S} = X_{N(j)}(\alpha^S)$ where $\alpha^S \in \{0,1\}^{N(j)}$ is the partial assignment defined by $\alpha^S_i = 1$ if $i \in S$ and $\alpha^S_i = 0$ if $i \in N(j) \setminus S$.
\end{itemize}
\end{proof}

\section{Proof of Theorems~\ref{thm:sa} and~\ref{thm:lasserre}}
\label{sec:proofs}
\begin{theorem}
[Restatement of Theorem~\ref{thm:sa}]
Let $G$ be a finite set, $k_{min} \geq 3$, and $\calp = \{P_1, \dots , P_l \}$ be a collection of predicates such that each $P_i \subseteq G^{k_i}$ supports a balanced pairwise independent distribution $\mu_i$. 
Let $(V, \calc)$ be an instance of CSP($\calp$) such that $\calc$ is $(s, 2 + \delta)$-boundary expanding for some $0<\delta \leq \frac{1}{4}$. 
Then, there exists a balanced solution to the $\frac{\delta s}{6k_{max}}$ rounds of the Sherali-Adams hierarchy that satisfies every constraint in $\calc$. 

\end{theorem}
\begin{proof}
The proof closely follows Theorem 4.3 of Georgiou, Magen, and Tulsiani~\cite{GMT09}. 
Their result, as a black-box, gives a solution to the Sherali-Adams hierarchy that satisfies all the constraints. There are two additional things that we need to check:
\begin{itemize}
\item More than one predicate: Unlike usual CSPs, our definition of Min-Ones($\calp$) allows to use more than one predicate, and predicates can have different arities.

\item Balanced solution: For our purposes, we need the solution to be balanced (i.e., $X_v(g) = \frac{1}{|G|}$ for all $v$ and $g$). 
\end{itemize}

The main part of their proof (Lemma 3.2) is robust to the two issues described above. As many technical parts of the proof can be used as a black-box, we sketch the high-level ideas of the proof and highlight the reason why it is robust to the two issues discussed above. We give the following additional definitions for a CSP-instance after removing some variables: 
Given an instance $(V, \calc)$ of CSP($\calp$) and a subset $S \subseteq V$, let $\calc(S)$ denote the set of all constraints that are entirely contained in $S$, namely: $\calc(S) := \{ C_i : E_i \subseteq S \}$). 
Let $(V \setminus S, \calc \setminus \calc(S))$ be the {\em instance after removing $S$}, namely: for each $C_i \in \calc \setminus \calc(S)$, the set $E_i$ is replaced by $E_i \cap (V \setminus S)$ and its predicate becomes the corresponding projection of $P_{t_i}$ on $G^{|E_i \cap (V \setminus S)|}$.

\paragraph{Expansion Correction.}
Let $S$ be a subset of $V$ and $\calc(S) = \{ C_i = (E_i, t_i) \}_{i = 1, \dots , m_S}$ be the constraints induced by $S$. Each predicate $P_{t_i}$ is associated with a balanced pairwise independent distribution $\mu_{t_i}$. Perhaps the most natural way to combine these distributions to define a local distribution on the assignments $\{ \alpha : S \rightarrow G \}$ is to take the (normalized) product of all the distributions, i.e.,
\[
\Pr_S[\alpha] =  (\prod_{i = 1}^{m_S} \mu_{t_i}(\alpha(e_{i, 1}), \dots , \alpha(e_{i, k_{t_i}}))) / Z_S,
\]
\[
Z_S = \sum_{\alpha : S \rightarrow G} \prod_{i = 1}^{m_S} \mu_{t_i}(\alpha(e_{i, 1}), \dots , \alpha(e_{i, k_{t_i}})).
\]
Call this distribution {\em canonical} for $S$. Clearly, any assignment $\alpha$ that has a positive probability will satisfy all constraints in $\calc(S)$. 

For any subset $S$, we can define the canonical local distribution. But generally the distributions will not be consistent (i.e., for some $S \subseteq S'$, the canonical distribution on $S$ might be different from the marginal distribution on $S$ obtained from the canonical distribution on $S'$). 
Since the canonical distribution on $S'$ induces a local distribution on any $S \subseteq S'$, it might be possible that the canonical distributions of carefully chosen sets are consistent and induce a local distribution for every set we are interested in.

Georgiou et al. \cite{GMT09} define the canonical distribution on some family $\barcals$ of sets that satisfies the following conditions:
\begin{itemize}
\item Any $\bar{S} \in \barcals$ satisfies $|\bar{S}| \leq \frac{s}{4}$. 
\item For any set $S \subseteq V$ with $|S| \leq \delta s / (6k_{max})$, there is an $\bar{S} \in \barcals$ such that $S \subseteq \bar{S}$. 
\item For any $\bar{S} \in \barcals$, the instance $(V \setminus \bar{S}, \calc \setminus \calc(\bar{S}))$, obtained by removing $\bar{S}$ and its induced constraints, is $(\frac{3}{4}s, \frac{8}{3} + \delta)$-boundary expanding. Recall that $(V \setminus \bar{S}, \calc \setminus \calc(\bar{S}))$ is different from 
the induced instance $(V \setminus \bar{S}, \calc(V \setminus \bar{S}))$. 
\end{itemize}
The existence of such an $\barcals$ is shown in Theorem 3.1 of~\cite{BGMT12}.\footnote{The corresponding theorem in the original version~\cite{GMT09} seems to have a minor error, so we here follow the final version of their work.} 

\paragraph{Consistent Distributions.}
The final local distributions $\{ X_S(\alpha) \}$ are defined as follows: for each $S$, find $\bar{S} \in \barcals$ that contains $S$, and use the canonical distribution defined on $\bar{S}$. It only remains to show that for any $\bar{S}, \bar{S'} \in \barcals$, their canonical distributions are consistent. 
The following lemma is the crucial part of~\cite{GMT09}.

\begin{lem}
\label{lem:gmt}
[Lemma 3.2 of~\cite{GMT09}]
Let $(V, \calc)$ be a CSP-instance as above and $S_1 \subseteq S_2$ be two sets of variables such that both $(V, \calc)$ and $(V \setminus S_1, \calc \setminus \calc(S_1))$ are $(t, 2 + \delta)$-boundary expanding for some $\delta \in (0, 1)$ and $|\calc(S_2)| \leq t$. Then for any $\alpha_1 \in G^{S_1}$, 
\[
\sum_{\alpha_2 \in G^{S_2}, \alpha_2(S_2) = \alpha_1} \Pr_{S_2}[\alpha_2] = \Pr_{S_1}[\alpha_1].
\]
\end{lem}

Applying Lemma~\ref{lem:gmt} two times 
(once with $(S_1,S_2) \leftarrow (\bar{S},\bar{S} \cup \bar{S'})$ and once with $(S_1,S_2) \leftarrow (\bar{S'},\bar{S} \cup \bar{S'})$), we conclude that both $\Pr_{\bar{S}}$ and $\Pr_{\bar{S'}}$ are marginal distributions of $\Pr_{\bar{S} \cup \bar{S'}}$, and hence should be consistent. 

We check the two issues which are not explicitly dealt in their paper. 
First, we note that $\Pr_S$ is defined as long as we have a distribution $\mu_i$ for each predicate $P_i$. The proof of Lemma~\ref{lem:gmt} only depends on the fact that each $\mu_i$ is balanced pairwise independent and not on any further structure of the predicates. 
Furthermore, predicates having different arities are naturally handled as long as we have $(t, 2+\delta)$-boundary expansion and pairwise independent distributions. 
Therefore, having more than one predicate with different arities does not affect the statement. Finally, we check that the resulting local distribution is balanced. Fix any variable $v \in V$ and let $\bar{S} \in \barcals$ be a set containing $v$. Applying Lemma~\ref{lem:gmt} with $S_1 \leftarrow \{ v \}$ and $S_2 \leftarrow \bar{S}$ ($\Pr_{\{v\}}$ is the uniform distribution on $G$ since $\{v\}$ does not contain any constraint), we get that the canonical distribution on $\bar{S}$ induces the uniform distribution on $G$ for $v$. 
\end{proof}

\begin{theorem}
[Restatement of Theorem~\ref{thm:lasserre}]
Let $G$ be a finite abelian group, $\calp = \{P_1, \dots , P_l \}$ be a collection of predicates such that each $P_i$ is a coset of a balanced pairwise independent subgroup of $G^{k_i}$ for $k_{min} \geq 3$. 
Let $(V, \calc)$ be an instance of CSP($\calp$) such that $\calc$ is $(s, 1 + \delta)$-expanding for some $\delta \leq \frac{1}{4}$. 
Then, there exists a balanced solution to the $\frac{s}{16}$ rounds of the Lasserre hierarchy that satisfies every constraint in $\calc$. 
\end{theorem}
\begin{proof}
The proof closely follows Theorem D.9 of Chan~\cite{Chan13},
which generalizes the work of Schoenebeck \cite{Schoenebeck08} and Tulsiani~\cite{Tulsiani09}.
His result, as a black-box, gives a solution to the Lasserre hierarchy that satisfies all the constraints. There are two additional things that we need to check:
\begin{itemize}
\item More than one predicate: Unlike usual CSPs, our definition of Min-Ones($\calp$) allows to use more than one predicate, and predicates can have different arities.

\item Balanced solution: For our purposes, we need the solution to be balanced (i.e., $||V_{v}(g)||_2^2 = \frac{1}{|G|}$ for all $v$ and $g$). 

\end{itemize}

Since these are immediate consequences of the previous results, 
instead of proving them in details, we describe the high-level ideas of the construction
while focusing on the points that we need to check.

\paragraph{Describing Each Predicate by Linear Equations.}
Let $\bbt$ be the unit circle in the complex plane. 
Given a finite abelian group $G$, 
let $\hat{G}$ be the set of characters (homomorphisms from $G$ to $\bbt$). 
$\hat{G}$ is again an abelian group (under pointwise multiplication) with the same cardinality as $G$.
The identity is the all-ones function $\1$, and the inverse of $\chi$ is $\frac{1}{\chi} = \bar{\chi}$, where $~ \bar{\cdot}~ $ indicates the complex conjugate. 

Consider $\widehat{G^V}$ which is isomorphic to $\hat{G}^V$. A character $\chi = (\chi_v)_{v \in V} \in \hat{G}^V$ is said to be $v$-relevant if $\chi_v \in \hat{G}$ is not the trivial character. The support of a character $\chi$ is defined to be $\supp(\chi) := \{ v \in V : \chi \mbox{ is } v\mbox{-relevant} \}$, and the weight of $\chi$ is $|\chi| := |\supp(\chi)|$. 

A linear equation is a pair $(\chi, z) \in \hat{G}^V \times \bbt$, 
and an assignment $f : V \rightarrow G$ satisfies $(\chi, z)$ if and only if $\chi(f) := \prod_{v} \chi_v(f(v)) = z$. 
Given a constraint $C_i = (E_i, t_i)$ where the predicate $P_{t_i}$ is a coset of a subgroup of $G^{k_i}$, there is a set of linear equations $L_i$ such that an assignment $f$ satisfies $C_i$ if and only if it satisfies all the linear equations in $L_i$. See Section D.1 of Chan~\cite{Chan13} for technical details. Since each predicate is equivalently formulated by a set of linear equations, having different predicates will not matter, as long as the linear equations have the desired properties.

\paragraph{Resolution Complexity.}
Given an instance of Min-Ones $(V, \calc)$ and the set $\call := \cup_i L_i$ of linear equations describing all the predicates, its width-$t$ resolution $\call_t$ is the smallest set satisfying the following:

\begin{itemize}
\item $\call \subseteq \call_t$. 
\item $(\chi, z), (\psi, y) \in \call_t$ and $|\chi \bar{\psi}| \leq t \Rightarrow (\chi \bar{\psi}, z \bar{y}) \in \call_t$. Say $(\chi \bar{\psi}, z \bar{y})$ is {\em derived} from $(\chi, z)$ and $(\psi, y)$. 
\end{itemize}

$\call_t$ is said to {\em refute} $\call$ if $(\1, z) \in \call_t$ with $z \neq 1$, and $\call_t$ is said to {\em fix} $v \in V$ if there exists $(\chi, z) \in \call_t$ with $\supp(\chi) = \{ v \}$. 

\begin{lem}
If $(V, \calc)$ is $(s, 1 + \delta)$-expanding for $\delta \leq 1/4$ and each predicate is a coset of a balanced pairwise independent subgroup, then $\call_{s / 8}$ can neither refute $\call$ nor fix a variable. 
\end{lem}
\begin{proof}
The proof is identical to that of Theorem 4.3 of Tulsiani, which Theorem D.8 of Chan follows, except that they only prove the lemma for refutation. We give the high-level ideas of the proof, pointing out that fixing a variable is also impossible.

Assume towards contradiction that $\call_t$ refutes $\call$ or fixes a variable, and let $(\chi^*, z^*) \in \call_t$ with $|\chi^*| \in \{ 0, 1 \}$. Without loss of generality, we can assume that $(\chi^*, z^*)$ is derived from $\{ (\chi_i, z_i) | 1 \leq i \leq m\}$, where each $(\chi_i, z_i)$ is derived only from $L_i$.
Let $S^* := \{i : \chi_i \neq \1 \}$ and $s^* := |S'|$. 
The crucial property they use is that $\chi_i$ with $i \in S^*$ has weight at least 3, which follows from the condition on predicates: Tulsiani requires a predicate to be a linear code of dual distance at least 3, and Chan requires it to be a balanced pairwise independent subgroup, which are indeed equivalent when $G$ is a finite field.

If $s^* \leq s$, since the instance is $(s, 1 + \delta)$-expanding, out of $\sum_{i \in S^*} |E_i|$ constraint-variable pairs $(i, e_{i, j})_{i \in S^*, 1 \leq j \leq k_{t_i}}$, at most $(2 + 2\delta)s^*$ pairs have another pair with the same variable. Since each $\chi_i$ with $i \in S^*$ has $|\chi_i| \geq 3$ and contributes 3 such pairs, at least $3s^* - (2 + 2\delta)s^* = (1 - 2\delta)s^*$ variables are covered exactly once by $\{ \supp(\chi_i) \}_{i \in S^*}$, making it impossible to derive any $(\chi, z)$ with $|\chi| < (1 - 2\delta)s'$. It shows that $s^* > s$. The original argument (Claim 4.4 of~\cite{Tulsiani09}) assumed that every predicate is of the same arity, but the above argument naturally adapted it to irregular arities.

Backtracking the derivations, we must have $(\chi', z') \in \call_{s/8}$, which is derived from $\frac{s}{2} \leq s' \leq {s}$ nontrivial characters from $L_i$'s (Claim 4.5 of Tulsiani). Similar expansion-minimum weight arguments again ensure that $|\chi'| > \frac{s}{8}$, which results in a contradiction. 
\end{proof}

\paragraph{Solution and Balance.}

Given that $\call_{s/8}$ does not refute $\call$, Theorem D.5 of~\cite{Chan13} ensures that there exists a solution $\{ V_{S}(\alpha) \}_{|S| \leq s/16, \alpha : S \rightarrow G}$ to the $s/16$ rounds of the Lasserre hierarchy that satisfies every constraint. 
Furthermore, one of his lemmas also proves that for every $v \in V$ and $g \in G$, $||V_{v}(g)||_2^2 = \frac{1}{|G|}$ using the fact that $\call_{s/8}$ does not fix any variable. 

\begin{lem}[Proposition D.7 of~\cite{Chan13}]
For $S \subseteq V$ with $|S| \leq s/16$,
let 
\[
H_S := \{ \beta | \beta : S \rightarrow G \mbox{ and } 
\beta \mbox{ satisfies every } (\chi, z) \in \call_{s/8} \mbox { with } \supp(\chi) \subseteq S
\}.
\]
For any $\alpha : S \rightarrow G$, 
\[
||V_{S}(\alpha)||_2^2 = \frac{\Indi[\alpha \in H_S]}{|H_S|},
\]
where $\Indi[\cdot]$ is the indicator function. 
\end{lem}
Combining all three parts above, we have a balanced solution to the $\frac{s}{16}$ rounds of the Lasserre hierarchy that satisfies every constraint. 
\end{proof}

\section{Properties of Random LDPC codes}
\label{sec:graphs}

\begin{lem}
[Restatement of Lemma~\ref{lem:mb_almost_regular}]
Consider the parity-check graph of a random $(d_v, d_c)$-LDPC code. With high probability, every vertex on the left (resp. right) will have degree either $d_v$ or $d_v - 2$ (resp. $d_c$ or $d_c - 2$).
\end{lem}
\begin{proof}
Let $M := n d_v = m d_c$. Fix a vertex $v$ on the left. In order to have at most $d_v - 2$ neighbors, $v$ needs to either have a neighbor with triple edges or two neighbors with double edges. The probability of the first event is at most by $m \cdot \binom{d_v}{3} \cdot \binom{d_c}{3} \cdot 3! \cdot \frac{1}{M(M-1)(M-2)} = O(\frac{1}{n^2})$. 
The probability of the second event is at most by $m^2 \cdot \binom{d_v}{4} \cdot (\binom{d_c}{2})^2 \cdot 4! \cdot \frac{1}{M(M-1)(M-2)(M-3)} = O(\frac{1}{n^2})$. By taking a union bound over all $v$, the probability that there exists a vertex with at most $d_v - 2$ different neighbors is $O(\frac{1}{n})$. The proof for the right side is similar. 
\end{proof}

\begin{lem}
[Restatement of Lemma~\ref{lem:mb_expansion}]
Given any $0 < \delta < 1/2$, there exists $\eta > 0$ (depending on $d_c$) such that the parity-check graph of a random $(d_v,d_c)$-LDPC code is $(\eta n, 1 + \delta)$-expanding with high probability.
\end{lem}

\begin{proof}
Let $k := d_c$. Fix a set $S$ of $s \leq \eta m$ vertices on the right for some $\eta > 0$ chosen later. Suppose that the degree of each vertex in $S$ is given. By the above lemma, with high probability, each degree is either $k$ or $k - 2$. Let $\bar{k}$ be the average degree of these $s$ vertices, and $\bar{c} = \bar{k} - 1 - \delta$. Fix a set $\Gamma$ of $\bar{c} s$ vertices on the left.

For a vertex $v \in S$ with degree $k'$, the probability that it has all $k'$ neighbors from $\Gamma$ is at most $(\frac{2\bar{c}s}{n})^{k'}$. If we condition that other vertices in $S$ have neighbors in $\Gamma$, this estimate only decreases. Therefore, the probability that the vertices in $S$ have neighbors only from the $\Gamma$ is at most $(\frac{2\bar{c}s}{n})^{\bar{k} s}$. Taking a union bound over $\binom{n}{\bar{c}s} \leq (\frac{ne}{\bar{c}s})^{\bar{c}s}$ choices of $\Gamma$, conditioned on any degrees of $S$, the probability of the bad event conditioned on any sequence of degrees of $S$ is at most 
\[
(\frac{2\bar{c}s}{n})^{\bar{k} s} \cdot (\frac{ne}{\bar{c}s})^{\bar{c}s}
\leq n^{(-1 - \delta)s} (k s)^{(1 + \delta)s} (2e)^{ks}.
\]

Taking a union bound over $\binom{m}{s} \leq \binom{n}{s} \leq (\frac{en}{s})^s$ choices for $S$, the probability that some set $S$ of size $s$ becomes bad is at most $(\frac{s}{n})^{\delta s} (k^{1 + \delta} (2e)^{k+1})^s$.
Let $\beta = k^{1+\delta} (2e)^{k+1}$ so that the above quantity becomes $(\frac{s}{n})^{\delta s} \beta^s = (\frac{s \beta^{1/\delta}}{n})^s$. When we sum this probability over all $s \leq \eta n$, we have
\[
\sum_{s=1}^{\eta n} (\frac{s \beta^{1/\delta}}{n})^{\delta s} 
= \sum_{s=1}^{\ln^2 n} (\frac{s \beta^{1/\delta}}{n})^{\delta s} + 
\sum_{s = \ln^2 n + 1} (\frac{s \beta^{1/\delta}}{n})^{\delta s}
\leq O(\frac{\beta^1}{n^{\delta}}\ln^2 n) + O((\eta \cdot \beta^{1/\delta})^{\delta \ln^2 n}).
\]
The first term is $o(1)$ for large $n$. The second term is also $o(1)$ for $\eta < 1/(\beta^{1/\delta})$.
\end{proof}

\section{More on Pairwise Independent Distributions}
\label{sec:sa_more}
\begin{lem}\label{le:imp_k_k_minus_1}
Let $G = \{0, \dots , k - 1 \}$ be a finite set. 
There is no balanced pairwise independent distribution $\nu$ on $G^k$ where every atom $(x_1, \dots, x_k)$ in the support has at least one $0$ coordinate. 
\end{lem}
\begin{proof}
Given $x = (x_1, \dots, x_k) \in G^k$, let $|x|$ be the number of 0's among $x_1, \dots, x_k$. The fact that $\mu$ is balanced implies $\E_{x \sim \mu}[|x|] = 1$, but the other requirement implies $|x| \geq 1$ for any $x$ in the support. Therefore, any $x$ in the support satisfies $|x| = 1$. Fix any $i \neq j$. If $x_i = 0$, $x_j$ cannot be 0 and $x_i$ and $x_j$ are not pairwise independent. 
\end{proof}

\begin{lem}\label{le:imp_k_q+1_even}
Let $G = \{0, \dots, k - 2 \}$ be a finite set for even $k$.
There is no balanced pairwise independent distribution $\nu$ on $G^k$ where every atom $(x_1, \dots, x_k)$ in the support has an odd number of zeros. 
\end{lem}
\begin{proof}
Assume for contradiction that such a $\mu$ exists. For odd $1 \leq i \leq k - 1$, let $a_i$ be the probability that the $(x_1, \dots, x_k)$ sampled from $\mu$ has exactly $i$ zeros. From balanced pairwise independence, they should satisfy the following set of inequalities: 
\begin{itemize}
\item Valid probability distribution: $\sum_{1 \leq i \leq k-1, i \mbox { odd}}a_i = 1$.
\item Balance: $\sum_{1 \leq i \leq k-1, i \mbox { odd}} a_i \cdot \frac{i}{k} = \frac{1}{k - 1} \Leftrightarrow 
\sum_{1 \leq i \leq k-1, i \mbox { odd}} i a_i  = \frac{k}{k - 1}$.
\item Pairwise independence: $\sum_{3 \leq i \leq k-1, i \mbox { odd}} a_i \cdot \frac{i(i-1)}{k(k-1)} = \frac{1}{(k - 1)^2}\Leftrightarrow 
\sum_{3 \leq i \leq k-1, i \mbox { odd}} i(i-1) a_i  = \frac{k}{k - 1}$.
\end{itemize}
Subtracting the first equation from the second, we get $\sum_{3 \leq i \leq k-1, i \mbox { odd}} (i - 1) a_i  = \frac{1}{k - 1}$. Subtracting $k$ times this equation from the third equation, we get $\sum_{3 \leq i \leq k - 1, i \mbox { odd}} (i - 1)(i - k) a_i = 0$, which is contradiction since all $a_i \geq 0$.
\end{proof}

\begin{lem}
[Restatement of Lemma~\ref{sa:predicates}]
Let $G = \{ 0, 1, \dots , q - 1 \}$ be a finite set. For the following combinations of arity values $k$ and alphabet size values $q$, each of the odd predicate and the even predicate supports a balanced pairwise independent distribution on $G^k$. 
\begin{itemize}
\item Even $k \geq 4$, $q = k - 2$.
\item Odd $k \geq 5$, $q = k - 3$.
\item Even $k \geq 6$, $q = k - 4$.
\end{itemize}
\end{lem}
\begin{proof}
We again construct each distribution by sampling $x \in \{0, 1\}^k$ first. $y = (y_1, \dots, y_k) \in G^k$ is given 
\begin{itemize}
\item For each $i$, if $x_ i = 0$, $y_i \leftarrow 0$. 
\item If $x_i \neq 0$, $y_i$ is chosen uniformly from $\{1, \dots, q - 1 \}$. independently. 
\end{itemize}
If $x$ is $\frac{q - 1}{q}$-biased and pairwise independent on $\{ 0, 1 \}^k$, it is easy to check that $y$ is balanced pairwise indepedent on $G^k$. From now on, we show how to sample the vector $x$ and prove that it satisfies the desired properties. 

\paragraph{Even $k \geq 4$, $q = k - 2$.}
We first deal with the odd predicate. 
Our strategy to sample $x$ is the following. Sample $r \in \{1, 3, k - 1\}$ with probabilty $a_1, a_3, a_{k-1}$ respectively. Sample a set $R$ uniformly from $\binom{\{ 1, 2, \dots, k \}}{r}$ and fix $x_i = 1$ if and only if $i \in R$. The probabilities $a_1, a_3, a_{k-1}$ should satisfy the following three equations.
\begin{itemize}
\item Valid probability distribution: $a_1 + a_3 + a_{k - 1} = 1$. 
\item $\frac{q - 1}{q}$-biased: $\frac{1}{k} a_1 + \frac{\binom{k - 1}{2}}{\binom{k}{3}}a_3 + \frac{k - 1}{k} a_{k - 1} = \frac{1}{k - 2} \Leftrightarrow a_1 + 3a_3 + (k - 1)a_{k - 1} = \frac{k}{k - 2}$.
\item Pairwise Independence: $\frac{\binom{k - 2}{1}}{\binom{k}{3}}a_3 + \frac{k - 2}{k} a_{k - 1} = (\frac{1}{k - 2})^2 \Leftrightarrow 6a_3 + (k - 1)(k - 2)a_{k - 1} = \frac{k(k-1)}{(k - 2)^2}$.
\end{itemize}
$a_1 = \frac{2k^3 - 13k^2 + 25k - 12}{2k^3 - 12k^2 + 24k - 16}, 
a_3 = \frac{k - 1}{2k^2 - 8k + 8}, a_{k - 1} = \frac{k - 3}{k^3 - 6k^2 + 12k - 8}$ is the solution to the above system. They are well-defined and nonnegative for $k \geq 4$. 

For the even predicate, we can choose $x$ as above, using $r \in \{ 0, 2, 4 \}$. 
\begin{itemize}
\item Valid probability distribution: $a_0 + a_2 + a_4 = 1$. 
\item $\frac{q - 1}{q}$-biased: $\frac{\binom{k - 1}{1}}{\binom{k}{2}}a_2 + \frac{\binom{k - 1}{3}}{\binom{k}{4}} a_4 = \frac{1}{k - 2} \Leftrightarrow 2a_2 + 4a_4 = \frac{k}{k - 2}$.
\item Pairwise Independence: $\frac{1}{\binom{k}{2}}a_2 + \frac{\binom{k - 2}{2}}{\binom{k}{4}} a_{4} = (\frac{1}{k - 2})^2 \Leftrightarrow 2a_2 + 12a_{4} = \frac{k(k-1)}{(k - 2)^2}$.
\end{itemize}
$a_0 = \frac{4k^2 - 23k + 32}{8k^2 - 32k + 32}, a_2 = \frac{2k^2 - 5k}{4k^2 - 16k + 16}, a_4 = \frac{k}{8k^2 - 32k + 32}$ is the solution to the above system. They are well-defined and nonnegative for $k \geq 4$.

\paragraph{Odd $k \geq 5$, $q = k - 3$.}
We can use the same framework as above, except that in every equation, the denominator of the RHS is changed from $k - 2 $ to $k - 3$.

For the even predicate, $a_0 = \frac{2k^2 - 17k + 36}{4k^2 - 24k + 36}, a_2 = \frac{k^2 - 4k}{2k^2 - 12k + 18}, a_4 = \frac{k}{4k^2 - 24k + 36}$ is the solution to
\begin{align*}
a_0 + a_2 + a_4 &= 1 \\
2a_2 + 4a_4 &= \frac{k}{k - 3} \\
2a_2 + 12a_4 &= \frac{k(k-1)}{(k-3)^2}.
\end{align*}
They are well-defined and nonnegative for $k \geq 5$.

For the odd predicate, we have that $a_1=\frac{k^3-8k^2+16k}{k^3-7k^2+15k-9},
a_3=\frac{k^2-4k}{k^3-9k^2+27k-27},
a_k = \frac{k^2-10k+27}{k^4-10k^3+36k^2-54k+27}$ is the solution to
\begin{align*}
a_1 + a_3 + a_k &= 1 \\
a_1 + 3a_3 + ka_k &= \frac{k}{k - 3} \\
6a_3 + k(k-1)a_k&= \frac{k(k-1)}{(k-3)^2}.
\end{align*}
They are well-defined and nonnegative for $k \geq 5$. 

\paragraph{Even $k \geq 6$, $q = k - 4$.}
We can use the same framework as above, except that in every equation, the denominator of the RHS is changed from $k - 3$ to $k - 4$.

For the even predicate, $a_0 = \frac{4k^2 - 45k + 128}{8k^2 - 64k + 128}, a_2 = \frac{2k^2 - 11k}{4k^2 - 32k + 64}, a_4 = \frac{3k}{8k^2 - 64k + 128}$ is the solution to
\begin{align*}
a_0 + a_2 + a_4 &= 1 \\
2a_2 + 4a_4 &= \frac{k}{k - 4} \\
2a_2 + 12a_4 &= \frac{k(k-1)}{(k-4)^2}.
\end{align*}
They are well-defined and nonnegative for $k \geq 6$.

For the odd predicate, we have 
$a_1 = \frac{2k^3-23k^2+75k-48}{2k^3-20k^2+64k-64}, 
a_3 = \frac{3k^2-19k+16}{2k^3-24k^2+96k-128}, 
a_{k - 1} = \frac{k^2-13k+48}{k^4-14k^3+72k^2-160k+128}$ 
to
\begin{align*}
a_1 + a_3 + a_k &= 1 \\
a_1 + 3a_3 + (k-1) a_k &= \frac{k}{k - 4} \\
6a_3 + (k-1)(k-2)a_k&= \frac{k(k-1)}{(k-4)^2}.
\end{align*}
They are well-defined and nonnegative for $k \geq 6$. 
\end{proof}

\section{More on Pairwise Independent Subgroups}
\label{sec:subgroup_more}

\begin{lem}
[Restatement of Lemma~\ref{lem:lasserre_2q}]
Let $q$ be a power of $2$ and $k = 2q$. There exists a subgroup of $\F_q^k$ such that every element in the subgroup contains an even number of $0$ coordinates. 
\end{lem}

\begin{proof}
Our subgroup $H'$ will be of the form $\{ (\alpha x + \beta y + \gamma z)_{(x, y, z) \in E} \}_{\alpha, \beta, \gamma \in F_q}$, for some subset $E \subseteq \F_q^3$ of $2q = k$ evaluation points. 
The set $E \subseteq \mathbb{F}_q^3$ is given by
$$ E := \{(1,a,a) : a \in \mathbb{F}_q\} \cup \{ (0,b,b+1): b \in \mathbb{F}_q \}. $$
Clearly, $|E| = 2q$. The lemma follows from Claim~\ref{le:even_num_roots_claim} and Claim~\ref{le:balance_pi_subgroup_claim} below. 
\end{proof}

\begin{claim}\label{le:even_num_roots_claim}
Every trivariate linear form $(\alpha x + \beta y + \gamma z)$  has either $0$, $2$, $q$ or $2q$ roots in $E$ (which are all even integers).
\end{claim}

\begin{proof}
Let $\psi_{\alpha,\beta,\gamma}$ be a fixed trivariate $\mathbb{F}_q$-linear form, for some $\alpha,\beta,\gamma \in \mathbb{F}_q$. Let $E_1 := \{(1,a,a): a \in \mathbb{F}_q\}$ and $E_2 := \{(0,b,b+1): b \in \mathbb{F}_q\}$. We distinguish two cases:
\begin{itemize}
\item Case $1$: $\beta+\gamma \neq 0$ in $\mathbb{F}_q$. Then, $\psi_{\alpha,\beta,\gamma}(1,a,a) = 0$ if and only if $a (\beta+\gamma) = - \alpha$, which is equivalent to $a = -(\beta+\gamma)^{-1} \alpha$. Hence, $\psi_{\alpha,\beta,\gamma}$ has exactly one root in $E_1$. Moreover, $\psi_{\alpha,\beta,\gamma}(0,b,b+1) = 0$ if and only if $b (\beta+\gamma) = - \gamma$, which is equivalent to $b = - (\beta+\gamma)^{-1} \alpha_3$. Hence, $\psi_{\alpha,\beta,\gamma}$ has exactly one root in $E_2$. So we conclude that in this case $\psi_{\alpha,\beta, \gamma}$ has exactly $2$ roots in $E = E_1 \cup E_2$.
\item Case $2$: $\beta + \gamma = 0$ in $\mathbb{F}_q$. Then, $\psi_{\alpha,\beta, \gamma}(1,a,a) = 0$ if and only if $a (\beta + \gamma) = - \alpha$, which is equivalent to $\alpha_1 = 0$. Hence, $\psi_{\alpha,\beta, \gamma}$ has either $0$ roots in $E_1$ (if $\alpha \neq 0$) or $q$ roots in $E_1$ (if $\alpha = 0$). Moreover, $\psi_{\alpha,\beta, \gamma}(0,b,b+1) = 0$ if and only if $b(\beta + \gamma) = - \gamma$, which is equivalent to $\gamma = 0$. Hence, $\psi_{\alpha,\beta, \gamma}$ has either $0$ roots in $E_2$ (if $\gamma \neq 0$) or $q$ roots in $E_2$ (if $\gamma = 0$). So we conclude that in this case $\psi_{\alpha,\beta, \gamma}$ has either $0$, $q$ or $2q$ roots in $E = E_1 \cup E_2$.
\end{itemize}
\end{proof}

\begin{claim}\label{le:balance_pi_subgroup_claim}
$H'$ is a balanced pairwise independent subgroup of $\F_q^k$.
\end{claim}

\begin{proof}
Applying Lemma \ref{lem:pi} with $d = 3$, it is enough to show that any two distinct vectors in $E$ are linearly-independent over $\mathbb{F}_q$. To show this, assume for the sake of contradiction that there exist $v_1 \neq v_2 \in \mathbb{F}_q$ and a scalar $\beta \in \mathbb{F}_q$ such that $v_2 = \beta v_1$. We distinguish three cases:
\begin{itemize}
\item $v_1, v_2 \in E_1$. Then, $v_1 = (1,a_1,a_1+1)$ and $v_2 = (1,a_2,a_2+1)$ for some $a_1 \neq a_2 \in \mathbb{F}_q$. Then, $v_2 = \beta v_1$ implies that $\beta = 1$ and hence $a_2 = a_1$, a contradiction.
\item $v_1, v_2 \in E_2$. Then, $v_1 = (0,b_1,b_1+1)$ and $v_2 = (0,b_2,b_2+1)$ for some $b_1 \neq b_2 \in \mathbb{F}_q$. Then, $v_2 = \beta v_1$ implies that $\beta = 1$ and $b_1 = b_2$, a contradiction.
\item $v_1 \in E_1$ and $v_2 \in E_2$. Then, $v_1 = (1,a,a)$ and $v_2 = (0,b,b+1)$ for some $a,b \in \mathbb{F}_q$. Then, $v_2 = \beta v_1$ implies that $\beta = 0$ and hence that both $b = 0$ and $b+1 = 0$, a contradiction.
\end{itemize}
\end{proof}

\section{Proof of Theorem~\ref{thm:hvc} for Hypergraph Vertex Cover}
\label{subsec:hvc}
The result for $k$-Hypergraph Vertex Cover will follow from the machinery and predicates that we constructed in Sections~\ref{sec:solns_des_struct} and~\ref{sec:ldpc}. We first restate Theorem~\ref{thm:hvc}.

\begin{theorem}[Restatement of Theorem~\ref{thm:hvc}]\label{thm:restatement_hvc}
Let $k = q+1$ where $q$ is any prime power. For any $\epsilon > 0$, there exist $\beta, \eta > 0$ (depending on $k$) such that a random $k$-uniform hypergraph with $n$ vertices and $m = \beta n$ edges, simultaneously satisfies the following two conditions with high probability. 
\begin{itemize}
\item The integral optimum of $k$-HVC is at least $(1 - \epsilon)n$.
\item There is a solution to the $\eta n$ rounds of the Lasserre hierarchy of value $\frac{1}{k - 1}n$. 
\end{itemize}
\end{theorem}

In the rest of this section, we prove Theorem~\ref{thm:restatement_hvc}. Fix $k$ such that $q = k - 1$ is a prime power. Given an instance of $k$-HVC, which is an instance of Min-Ones($\{ P_{\vee} \}$), we stretch the domain from $\{ 0, 1 \}$ to $\F_q$ by the map $\phi : \F_q \rightarrow \{ 0, 1\}$ with $\phi(0) = 1$, $\phi(g) = 0$ for $g \neq 0$. Then the corresponding predicate $P'_{\vee} \subseteq \F_q^k$ is a tuple of $k$ elements from $\F_q^k$ that has at least one zero. We show that $P'_{\vee}$ contains a pairwise independent subgroup $H'$ of $\F_q^k$. 
Indeed, we use the same $H'$ that was used for the odd predicate for random LDPC codes, i.e., $H' := \{ (\alpha x + \beta y )_{(x, y) \in E} \}_{\alpha, \beta \in F_q}$ where $E := \{ (0, 1) \} \cup \{ (1, a) \}_{a \in \F_q}$.
In Section~\ref{subsec:lasserre}, we proved that $H'$ is balanced pairwise independent and always has an odd number of zeros when $k$ is odd. Here we allow $k$ to be even so this is not true, but we still have that any element of $H'$ has at least one zero (indeed, the only element in $H$ that does not have exactly one zero is $(0, 0, \dots , 0)$, which has $k$ zeros). This constructs the desired predicate for $P'_{\vee}$. Given this predicate, the same technique of stretching the domain, constructing a Lasserre solution by Theorem~\ref{thm:lasserre}, and collapsing back the domain using Lemma~\ref{lem:collapsing} gives a solution to the Lasserre hierarchy that is $\frac{1}{k - 1}$-biased. Lemma~\ref{lem:hvc_expansion} below, which ensures that random $k$-uniform hypergraphs have a large integral optimum and are highly expanding for some fixed number of hyperedges, concludes the proof of Theorem~\ref{thm:restatement_hvc}.

\begin{lem}
Let $k \geq 3$ be a positive integer and $\epsilon, \delta > 0$. 
There exists $\eta \leq \beta$ (depending on $k$) such that a random $k$-uniform hypergraph $(V, E)$ with $\beta n$ edges, where each edge $e_i$ is sampled from $\binom{V}{k}$ with replacement, has the following properties with high probability.
\begin{itemize}
\item It is $(\eta n, k - 1 - \delta)$-expanding.
\item Every subset of $\epsilon n$ vertices contains a hyperedge. Therefore, the optimum of $k$-HVC is at least $(1 - \epsilon)n$. 
\end{itemize}
\label{lem:hvc_expansion}
\end{lem}

\begin{proof}
The proof uses standard probabilistic arguments and can be found in previous works~\cite{ACGKTZ10, Tulsiani09}. 
Fix a subset $S \subseteq V$ of size $\epsilon n$. The probability that one hyperedge is contained in $S$ is 
\[
\frac{\binom{\epsilon n}{k}}{\binom{n}{k}} \geq \frac{(\epsilon n / k)^k}{(en/k)^k} = (\epsilon / e)^k.
\]
The probability that $S$ does not contain any edge is at most 
\[
(1 - (\epsilon / e)^k)^{\beta n} \leq \exp(-(\epsilon / e)^k \beta n).
\]
Since there are $\binom{n}{\epsilon n} \leq (e / \epsilon)^{\epsilon n} = \exp(\epsilon n(1 + \log(1/\epsilon)))$ choices for $S$, if $\beta > (e/\epsilon)^k$, with high probability, every subset of $\epsilon n$ vertices contains a hyperedge. 

Now we consider the probability that a set of $s$ hyperedges contains at most $cs$ variables, where $c = k - 1 - \delta$. This is upper bounded by
\[
\binom{n}{cs} \cdot \binom{\binom{cs}{k}}{s} \cdot s! \binom{\beta n}{s} \cdot \binom{n}{k}^{-s},
\]
($\binom{n}{cs}$ for fixing variables to be covered, $\binom{\binom{cs}{k}}{s}$ for assigning them to $s$ hyperedges, $s! \binom{\beta n}{s}$ for a set of $s$ hyperedges) which is at most 
\[
(s/n)^{\delta s}(e^{2k+1 - \delta} k^{1 + \delta} \beta)^s
\leq (s/n)^{\delta s}\beta^{5s} = (\frac{s \beta^{5/\delta}}{n})^{\delta s}.
\]
By summing the probability over $s = 1, \dots, \eta n$, the probability that it is not $(\eta n, k - 1 - \delta)$-expanding is
\[
\sum_{s=1}^{\eta n} (\frac{s \beta^{5/\delta}}{n})^{\delta s} 
= \sum_{s=1}^{\ln^2 n} (\frac{s \beta^{5/\delta}}{n})^{\delta s} + 
\sum_{s = \ln^2 n + 1} (\frac{s \beta^{5/\delta}}{n})^{\delta s}
\leq O(\frac{\beta^5}{n^{\delta}}\ln^2 n) + O((\eta \cdot \beta^{5/\delta})^{\delta \ln^2 n}).
\]
The first term is $o(1)$ for large $n$. The second term is also $o(1)$ for $\eta < 1/(\beta^{5/\delta})$. 
\end{proof}

\end{document}